\documentclass{article} %
\usepackage{iclr2026_conference_unanon,times}

\usepackage[utf8]{inputenc} %
\usepackage[T1]{fontenc}    %
\usepackage{hyperref}       %
\usepackage{url}            %
\usepackage{booktabs}       %
\usepackage{amsfonts}       %
\usepackage{nicefrac}       %
\usepackage{microtype}      %
\usepackage{amsmath, amssymb}
\usepackage{subcaption}
\usepackage{hyperref}
\usepackage{url}
\usepackage{enumitem}

\usepackage{amsthm}

\usepackage[dvipsnames]{xcolor} %
\usepackage{xcolor,colortbl} %

\usepackage{algorithm}
\usepackage{algpseudocode}
\usepackage{proof} %
\newcommand{\blank}{{-}}

\newtheorem{theorem}{Theorem}[section]
\newtheorem{lemma}[theorem]{Lemma}
\newtheorem{proposition}[theorem]{Proposition}

\theoremstyle{definition}
\newtheorem{definition}[theorem]{Definition}
\newtheorem{remark}[theorem]{Remark}

\renewcommand{\eqref}[1]{Equation~\ref{eq:#1}}

\newcommand{\secref}[1]{Section~\ref{sec:#1}}
\newcommand{\secstworef}[2]{Sections~\ref{sec:#1} and~\ref{sec:#2}}

\newcommand{\appref}[1]{Appendix~\ref{app:#1}}

\newcommand{\figref}[1]{Figure~\ref{fig:#1}}
\newcommand{\figstworef}[2]{Figures~\ref{fig:#1} and~\ref{fig:#2}}

\newcommand{\tabref}[1]{Table~\ref{tab:#1}}
\newcommand{\tabstworef}[2]{Tables~\ref{tab:#1} and~\ref{tab:#2}}

\renewcommand{\algref}[1]{Algorithm~\ref{alg:#1}}

\newcommand{\thmref}[1]{Theorem~\ref{thm:#1}}
\newcommand{\lemref}[1]{Lemma~\ref{lem:#1}}

\newcommand{\propref}[1]{Proposition~\ref{prop:#1}}

\newcommand{\defref}[1]{Definition~\ref{def:#1}}

\usepackage[colorinlistoftodos,prependcaption,textsize=tiny,textwidth=14mm]{todonotes}

\DeclareRobustCommand{\DE}[3]{#3}
\DeclareRobustCommand{\VAN}[3]{#3}

\title{What Was Said, Not What Was `Thought': Type-6 Logic for CoT Verification}

\author{%
  Adrian de Wynter\\
  Microsoft and the University of York\\
  \texttt{adewynter@microsoft.com}
}

\begin{document}
\maketitle
\begin{abstract}
We introduce Type-6 logic, a variant of dynamic epistemic logic augmented with two operators (uncertainty and recurrence), designed to model the inferential dynamics of contemporary large language model (LLM) chain-of-thought (CoT) reasoning. 
Type-6 accounts for common LLM reasoning pathologies such as unlicensed revision, enthymemes, loopbacks, and unverifiable/incorrect claims. 
We propose a verifier based on Type-6 logic that builds a graph out the trace, and checks it against Type-6's axioms and inference rules. 
We evaluate our framework on LLM-generated CoTs four splits spanning formal and informal reasoning. 
Our verifier detects structurally unsound reasoning steps that surface-level heuristics miss, and allows for easy visualisation of the model's reasoning process. 
In our corpus, our verifier shows that derived contradiction is the most common hard-fail category in CoT, and that only about 3\% of the propositions of a trace have impact on the final derivation. 
Ablation studies show that other verification methods (LLMs-as-judges, other neurosymbolic approaches, etc.) cannot be considered interchangeable: for example, agreement between LLMs-as-judges and LINC is $\kappa \approx 0.034$, and this persists within a method across underlying models. 
Type-6, however, is the most agreed-with method amongst the ones we tested. 
We prove our verifier runs on average-case linear time; and release our logic specification and artefacts. 

\end{abstract}

\section{Introduction}\label{sec:introduction}

Chain-of-thought (CoT) prompting is now a standard way to induce complex problem-solving capabilities in large language models (LLMs; \citealt{wei2022chain}). 
This, in theory, means that an LLM's reasoning steps are auditable through their traces. 
However, CoT lengths mean that they are often intractable for humans to verify. 
Thus, a body of work has been developed to create automated methods, 
such as process reward models (PRMs; \citealt{lightman2023lets,uesato2022solving}), 
prompts \citep{wang2023selfconsistency,madaan2023self}, 
and coherence metrics such as ROSCOE \citep{golovneva2023roscoe}. 
These methods largely conflate two questions: 
whether a trace \emph{arrives at} a correct answer, 
and whether its steps \emph{cohere with one another}. 
Indeed, it is known that traces often reach the right answer through incoherent steps \citep{zheng2025tmlr-curse}. 

We argue that the tokens a model emits do not carry provable truth conditions \emph{for the model itself}: 
we cannot ascribe knowledge or belief to the generator in any meaningful sense. 
However, \textbf{the statements and conclusions assembled from those tokens do carry truth} conditions \textit{for the reader}. 
Verification, on this view, lives on the interpretive side of the trace: 
we do not check what the model knows, but what the trace commits to when read as a natural-language reasoning artefact. 
Because CoT verification itself is ill-defined \citep{zaman-srivastava-2026-chain}, other methods--usually designed for a specific application--often miss this and other pathologies. 
For example, a trace may commit to $p$ and then, without derivation, commit to $\neg p$--an \textit{unlicensed revision}. 
It could also conclude $a$ from premises the chain never establishes (an \emph{unverifiable derivation}; distinct from contradiction). 
Or, it may state a conclusion without ever attempting to derive it (\emph{reasoning avoidance}), or leave key premises doubted or unresolved at the point of concluding. 
Consider:
\begin{quote}
\small
\texttt{K p1}     : I know memories are reliable evidence of the past.\\
\texttt{K NOT p1} : Actually, memories could be planted.\\
\texttt{p2}       : If memories were planted, we can't trust that yesterday existed.\\
\texttt{THEN a}   : So the world was created five minutes ago.
\end{quote}
Both the unlicensed knowledge-flip on $p_1$ (a hard contradiction), 
and the conclusion $a$ drawn from an antecedent whose grounding depends on an unresolved
conditional (an unverifiable derivation), would be missed by answer-correctness and NLI scorers. 
These are semantic pathologies of reasoning-as-written, and require a more nuanced approach accounting for the fact that CoT traces loop, meander, jump from point to point, and do not necessarily maintain a coherent or sound narrative--regardless of the answer's final correctness. 

For this we introduce \textbf{Type-6 logic,} a trace-driven update-semantic logic for CoT verification. 
Type-6 assigns each proposition a state $(\tau, \sigma, \delta)$--truth value, commitment strength (knowledge, belief, or uncommitted), and doubt. 
It specifies how modal operators, connectives, and structural markers (loopbacks, pivots, doubt) transform that state. 
We use Type-6 to build a CoT neurosymbolic verifier over a multigraph reconstruction of the trace. 
This verifier applies Type-6 rules to the reconstruction, and produces three reports: 
a structural \emph{coherence} report over all reasoning routes from opening to conclusion, 
a sequential \emph{verification} report logging contradictions and quality signals in order, 
and a \emph{per-path} report exposing where routes disagree or under-resolve. 
The design is paraconsistent by construction: 
contradictions are logged rather than suppressed, and the agent's commitments are always retained.
Crucially, Type-6 and its verifier require no gold labels, nor subscribe to a specific concept of CoT verification. Our contributions are:

\begin{itemize}[leftmargin=*]
\item \textbf{Type-6 logic}: a formal trace-driven update semantics for CoT, with a four-valued truth space, 
an explicit doubt flag, Kleene-style back-propagation over logical constraints, 
and a licensed-revision rule that distinguishes epistemic updates from unlicensed reversals. 
\item \textbf{A contradiction taxonomy}: hard and soft contradictions, 
residual signals (unresolved doubt or unknowability at the conclusion), 
a trace-level reasoning-avoidance signal, and cross-path disagreement; 
each mapped to concrete verifier output.
\item \textbf{A CoT verifier}: a single-pass graph construction equipped with Type-6 logic, 
producing reports with statement-level provenance.
\end{itemize}

We test our verifier against four common methods--LLMs-as-judges, ROSCOE, LINC \citep{olausson2023linc}, and PRMs--on a human-annotated corpus of 996 formal and informal reasoning traces. 
We find that none of these methods are interchangeable, since their failure decisions disagree with one another. 
However, Type-6 is the most-agreed-with method amongst them, correlating with LLM-judges ($\rho = 0.25$) and PRMs ($\rho = 0.31$) more strongly than baseline correlates with the other. 
We also find with Type-6's verifier deeper insights on traces, such as how many propositions lie on the main derivation chain on average (3\%). 
In all, Type-6 allows for comprehensive, tuneable, and auditable CoT trace verification.

\section{Background and Related Work}\label{sec:related}

\textbf{Scoring chain-of-thought.} Outcome-based scorers train on final-answer correctness \citep{cobbe2021training}, or aggregate over samples \citep{wang2023selfconsistency}. 
They treat the trace as a means to the answer, scoring two traces identically when one is coherent and the other is not. PRMs instead score individual steps against human annotations \citep{uesato2022solving,lightman2023lets}. They are successful at mathematical reasoning \citep{wang2024mathshepherd}, but require a trained model and do not generalise well beyond their domain. 
Judgement-based scorers prompt a model to evaluate another's output \citep{madaan2023self,shinn2023reflexion,zheng2023judging,li2023alpacaeval}. 
They are flexible, but carry known position and style biases, and are hard to audit at the step level. 
Reference-free coherence methods like ROSCOE use NLI and embedding metrics to the trace, and have been extended to reasoning-graph metrics \citep{prasad2023receval} and faithfulness probes \citep{lanham2023measuring,turpin2023language}. 
Our work is close to the latter two, but its neurosymbolic nature recovers a parseable structure. 
It also reports Type-6-based contradictions with statement-level provenance, as opposed to aggregate similarity.

\textbf{Neurosymbolic and structured-intermediate methods.} 
Other research converts natural-language reasoning into a formal representation and reasons over it symbolically. 
The most prominent is LINC \citep{olausson2023linc}. It translates premises and conclusions into a first-order logic (FOL) and calls a theorem prover. 
It has had success in arithmetic \citep{gao2023pal,chen2022programofthoughts} and constraint satisfaction \citep{ye2023satlm} tasks. 
These approaches build an intermediate structure similar to ours, but typically are restricted to FOL. 
More importantly, they do not cover--by design--informal reasoning, modal commitments, or approximately-correct derivations. 
Type-6 preserves these signals, in addition to structural markers and per-path decompositions. %

\textbf{Consistency, hallucination, and self-verification.} 
It is known that CoT as a formal reasoning process is flawed \citep{zheng2025tmlr-curse,yuehhan2026reasoningmodelsstrugglecontrol,turpin2023language,lanham2023measuring}; also see \citet{zaman-srivastava-2026-chain}. 
Approaches hence attempt to check whether the LLM itself is consistent via self-checking prompts \citep{wang2023selfconsistency,miao2024selfcheckgpt}, 
verifier-generator loops \citep{weng2023large}, 
or entailment probes. 
These target the generating LLM, while our verifier targets the generated artifact (trace). 
This distinction determines what the score \textit{means}: 
when the same trace comes from different generators, or there is no access to the generator, 
only the artifact-level reading is available. 
Hallucination detection \citep{ji2023survey,li2023halueval} shares this framing but focuses on factual accuracy rather than inferential structure.

\textbf{Formal and informal foundations.} 
Reading a reasoning trace as an object of interpretation, distinct from any agent's epistemic state, has roots in informal logic--typically defined as `anything that is not formal logic' \citep{perelman,WaltonInformal}; also see \cite{AnthonyBlair2015}. 
\citet{toulmin1958uses} argues that everyday reasoning is enthymematic, and \citet{walton2008informal} extends this with defeasible argumentation schemes. 
This is often applied to argument mining \citep{lawrence2019argument} and argumentation frameworks \citep{dung1995acceptability}. 
Type-6 is more classical, and closer to dynamic epistemic logic (DEL) \citep{baltag1998logic,vanditmarsch2007del}. 
It retains modal operators for knowledge and belief, but excludes multi-agent constructs (e.g., announcements) and replaces Kripke-frame semantics with an update semantics indexed by the trace. 
Still, it remains parametric on a specific system, such as S5 or KD45, since \textbf{a CoT does not need to conform to an agent's epistemic profile}, nor it always licenses a coherent Kripke frame. 
Type-6 also maintains a four-valued discrete truth space $\{T, U_c, U_k, F\}$, to avoid ascribing a subjective credence to the agent. 
This follows \citet{belnap1977useful}, 
, with $U_c$ playing the role of `no information'; 
and $U_k$ that of apending free variable. 
Back-propagation uses strong Kleene semantics \citep{kleene1952introduction} extended so that $U_c$ operands abstain. 
Paraconsistency (i.e., tolerating local contradiction without global explosion) is by design: 
Type-6 logs conflicting commitments rather than deriving arbitrary consequences, 
closer to the annotation-and-log style of \citet{arieli1998reasoning} than to dialetheic systems \citep{priest2007paraconsistent}. 
This logging means it also admits enthymemes, since its revision rules allow implicit derivations via back-propagation. 

\section{Type-6 Logic}\label{sec:type6}

\subsection{Operators and State}\label{sec:t6-informal}

Type-6 is a discrete, trace-driven, update-semantic logic meant for CoT verification and audit. 
It emphasises the fact that CoT tokens do not carry provable truth conditions for the generating LLM, 
but the statements assembled from them do, for the interpreter. 
Verification is then at the level of interpretation. 
This is because Type-6 logic is dynamic: a per-proposition state $\Sigma(\pi) = (\tau, \sigma, \delta)$ assigns a truth value
$\tau \in \{T, U_c, U_k, F\}$ (true, unclear/unparseable, unknowable/unresolved, false), 
a commitment strength $\sigma \in \{\bot, B, K\}$--where $K$ and $B$ are modal operators, discussed below--and a doubt flag $\delta \in \{0, 1\}$ (no/yes). 
It evolves as the trace is walked, and update rule violations are logged (\secref{t6-updates-failures}). 
This logging is the verifier in disguise. We discuss it in \secstworef{type6}{eval}. 

$\Sigma$ evolves via two operator types: \textbf{modal} and \textbf{structural}, in addition to the standard operators and modifier ($\mathrm{NOT}$) from FOL (\tabref{t6-ops}). 
The modal operators $B$ (belief) and $K$ (knowledge) capture the strength of an agent's commitment to a proposition. 
This separates statement-level content (what is said, or asserted) from modal commitment (how strongly). 

The structural operators $\mathord{?}$ (doubt), $R$ (loopback), and $N$ (pivot) govern how the trace revisits, questions, and organises its content without committing to propositional values. 
They are modelled after the standard CoT discourse. 
$N$ and $R$ do not change $\Sigma$, but change the logging: 
$N$ preserves an open constraint chain across a pivot, and $R$ marks the next commitment on a proposition as a licensed revision (\secref{t6-updates-failures}).
$\mathord{?}$ is the only operator allowed to alter $\delta$: 
doubt is not an implicit consequence of belief, so $B\,\pi$ (a positive commitment) and $B\,\pi$ followed by $\mathord{?}\,\pi$ (doubted belief) are distinct configurations. 

\begin{table}[t]
\centering
\small
\setlength{\tabcolsep}{4pt}
\begin{tabular}{llll}
\toprule
Operator & Agent-side read & Verifier-side update & Notes \\
\midrule
$K$ & Claims to know $\pi$ & $(\alpha,\, K,\, 0)$ & Requires $\alpha$$\in$$\{T$,$F\}$; clears $\delta$ \\
$B$ & Claims to believe $\pi$ & $(\alpha,\, B,\, \blank)$ & Keeps $\delta$; allows $\alpha$ $=$ $U_k$ \\
$\mathord{?}$ & Doubts its target & $(\blank,\, \blank,\, 1)$ & Only operator setting $\delta$=$1$ \\
$R$ & Revisits $\pi$ & No state change & licenses following revision \\
$N$ & Pivots the discourse & No state change & Inert \\
$\mathrm{NOT}$ & Negates next operator & Flip update's polarity & Folds into $K$/$B$/connective \\
$\mathrm{IF}$ & Creates antecedent & Opens an $\mathrm{IF}$-$\mathrm{THEN}$ chain & Ends any prior open chain \\
$\mathrm{THEN}$ & Sequential composition & Introduces constraint $r$ & Ends chain; triggers bprop \\
$\mathrm{AND}, \mathrm{OR}$ & Conjunction, disjunction & Extends the current chain & Per-operand polarity w/ $\mathrm{NOT}$ \\
(bare) & Assert $\pi$, no $K$,$B$ claim & $(\alpha,\, \blank,\, \blank)$ & Updates $\tau$; leaves $\sigma, \delta$ \\
\bottomrule
\end{tabular}
\caption{Type-6 operators and their induced update on $\Sigma(\pi) = (\tau, \sigma, \delta)$. 
$\alpha$ is the annotation on the current statement; $\blank$ denotes preservation. 
Chains are either open (via $\mathrm{IF}$), extended, or closed (via $\mathrm{THEN}$). 
The latter triggers backprop by introducing a constraint $r$. 
$N$ being inert means it preserves the open constraint chain. 
Modal operators ($B, K$) commit, but (bare) updates $\tau$ without committing. 
Structural operators ($\mathord{?}, R, N$) direct the flow.}\label{tab:t6-ops}
\end{table}

\textbf{Trace reading.} 
In Type-6, operators are read left-to-right, each scoping over the stack to its right within its statement. 
To facilitate reading a trace in Type-6, we encode assertions as the \emph{bare} case: a proposition asserted without $K$ or $B$ updates $\tau$ but does not change $\sigma$. 
This left-to-right read makes $\mathord{?}\,K\,\pi$ (`do I know $\pi$?') and $K\,\mathord{?}\,\pi$ (`I know I doubt $\pi$') distinct: 
the first commits $K$ on $\pi$ and flags doubt \textit{on the commitment}, 
and the second flags doubt \textit{on $\pi$} then commits $K$ on the doubt-state. 
At a first glance, doubted knowledge ($K\,\mathord{?}\,\pi$) and belief ($B\,\pi$) might look similar, 
since both yield a value one should not rely on without further support. 
However, these are semantically distinct, and $(\sigma, \delta)$ preserves this distinction. 
A worked example is in \appref{t6-nested}.

\subsection{Language and Update Rules}\label{sec:t6-formal}

Type-6's full description, constraint accumulator, Kleene connective tables, and axiomatic system and metatheory are in \appref{t6-axioms}. 
From before, remark that every commitment made by an agent can be logged, even when a later derivation suggests it is wrong. 
Indeed, under Type-6's update and language rules, when a back-propagated constraint assigns $\pi$ a value conflicting with the agent's committed value, it does not fail. 
Instead, it yields a derived contradiction but retains the agent's value. 
In turn, this agent-wins policy allows flagging `confidently incorrect' outputs--a well-known pathology in CoT.

\textbf{Language.}
Let $\mathcal{P}$ be a countable set of proposition symbols. 
The language $\mathcal{L}_{\text{Type-6}}$ consists of \emph{statements} in two classes:
\begin{align}
\text{Unary} &\quad s = \omega_1 \cdots \omega_k\, \pi : \varsigma,
  \label{eq:unary-stmt}\\
\text{Connective-led} &\quad s = \beta\, \omega_1 \cdots \omega_k\, \pi : \varsigma,
  \label{eq:conn-stmt}
\end{align}
where each $\omega_i \in \Omega_u = \{B, K, \mathord{?}, R, N, \mathrm{NOT}\}$, 
$\beta \in \Omega_c = \{\mathrm{THEN}, \mathrm{AND}, \mathrm{OR}, \mathrm{IF}\}$, 
$\pi \in \mathcal{P} \cup \{\varnothing\}$ is an optional proposition, 
and $\varsigma$ is an uninterpreted natural-language surface form. 
A connective-led statement binds $\pi$ to the proposition of the preceding statement, resolving the connective sequentially. 
A \emph{trace} is then a finite sequence $R = (s_1, \dots, s_n)$.

\textbf{Truth space and state.} The truth space $\mathcal{V} = \{T, U_c, U_k, F\}$ is ordered by resolvedness, $U_c \sqsubseteq U_k \sqsubseteq T, F$, with $T, F$ incomparable. 
$U_c$ (unclear) marks content the interpreter cannot parse, and 
$U_k$ (unknowable) marks content undetermined but resolvable in principle. 
The state is a partial function $\Sigma : \mathcal{P} \to \mathcal{V} \times \mathcal{C} \times \{0,1\}$ with $\mathcal{C} = \{\bot, B, K\}$. 
The initial state $\Sigma_0$ assigns $(U_k, \bot, 0)$ to every proposition in the trace.

\textbf{Updates.}
The modal updates commit a value and a strength,
\begin{align}
K\,\pi &: \Sigma' = \Sigma[\pi \mapsto (\alpha,\, K,\, 0)],\quad \alpha \in \{T, F\},
  \label{eq:K-update}\\
B\,\pi &: \Sigma' = \Sigma[\pi \mapsto (\alpha,\, B,\, \blank)],
  \label{eq:B-update}
\end{align}
where $\alpha$ is the annotation on the current statement. 
A $K$-commitment with $\alpha$$\in$$\{U_c, U_k\}$, or a $B$-commitment with $\alpha$=$U_c$, is a \emph{modal mismatch} (\secref{t6-updates-failures}). 
The $K$-update clears $\delta$ (knowledge is decisive) while $B$ preserves it (belief tolerates residual doubt). 
A \emph{bare} statement asserting $\pi$ without a modal operator updates $\tau$ and ignores $\sigma$. 
Doubt sets $\delta$=$1$ and preserves the rest; 
$N$ and $R$ do not alter $\Sigma$. 
Connective-led statements do not act on $\Sigma$ but extend or close a \emph{constraint chain}: 
a $\mathrm{THEN}$ closes the chain into a constraint $r$$\equiv$$\phi(p_1, \dots, p_k)$, which is added to a persistent constraint set and triggers back-propagation.

\textbf{Back-propagation} over the constraint set uses strong Kleene three-valued semantics \citep{kleene1952introduction}. 
We extend them so that a $U_c$ operand makes the enclosing derivation \emph{abstain} rather than propagate (\appref{t6-connectives}). 
Given some $r$$\equiv$$\phi(p_1, \dots, p_k)$ and the current state, 
back-prop either \emph{pins} an unresolved $\tau(r)$ to the derived value, 
logs a \emph{derived contradiction} when the derived value conflicts with a committed $\tau(r)$ (retaining the agent's value), 
or logs an \emph{unverifiable derivation} when the operands do not \textit{force} the committed target. 
Events are logged only on a constraint's status transition, so a single unresolved constraint does not flood the report.

\subsection{Updates and Failure Modes}\label{sec:t6-updates-failures}

\textbf{Revisions and reasoning avoidance.} 
An agent may revisit their prior beliefs, as well as skip the evaluation of certain statements. 
In other words, commitments are defeasible: an earlier value on $\pi$ may be overwritten, provided the overwrite is \emph{licensed}. 
A sub-trace between two opposite-valued commitments on $\pi$ is a \emph{revision segment} iff it contains 
(i) a $\mathsf{jump}$ or $\mathsf{loopback}$ edge ,and 
(ii) a constraint targeting $\pi$ (explicit or back-propagated). 
A revision across a revision segment is a \emph{licensed revision} and emits a revision event; 
without one, it is a $K$-$K$ contradiction (modal case) or a bare-reassertion conflict (bare case). 
The two-part condition is core to Type-6's enthymematic characteristics: 
implicit warrants are permitted, but the verifier requires positive evidence that revision was worked for. 
Note that licencing captures effort, not correctness: a licensed revision may still be wrong on the merits, which is a matter for downstream evaluation. 
Revision segments have trace analogues: 
if no constraint targets the sink proposition $a$, even transitively, the trace is \emph{reasoning-avoidant}: 
it announces a conclusion without attempting to derive it. 
This is a hard trace-level signal, since the absence of any conclusion-directed derivation is a property of the trace rather than of any route through it. 
Full definitions of the licencing window and its interaction with same-value re-assertion are in \appref{t6-axioms}.

\textbf{Contradictions and quality signals.} 
Parsing the logs yields two signals: \emph{contradictions}, which flag update-rule violations, and \emph{quality signals}, which flag rhetorical patterns impacting coherence without violating the logic. 
A summary of these is in \tabref{t6-contra}. 
Hard contradictions are direct violations of the update or back-propagation rules; 
soft contradictions are weaker inconsistencies that may reflect legitimate epistemic moves under some readings, and are flagged rather than blocking. 
Residual signals (unresolved doubt or unknowability at the sink) are properties of the terminal state and are elevated to hard when the residual proposition lies on the direct implication chain to the answer. %

\begin{table}[t]
\centering
\small
\begin{tabular}{lll}
\toprule
Category & Trigger & Severity \\
\midrule
$K$-$K$ contradiction & $K\,\pi{=}v_1$; $K\,\pi{=}v_2$; no revision segment & $^\dagger$Hard  \\
Derived contradiction & Back-prop yields $v' \neq v$ for committed $\pi$ & $^\dagger$Hard \\
Modal mismatch $(U_c)$ & $K\,\pi$ or $B\,\pi$ with $\alpha = U_c$ & $^\dagger$Hard \\
Modal mismatch $(U_k/K)$ & $K\,\pi$ with $\alpha = U_k$ & $^\dagger$Hard  \\
Reasoning-avoidance & No constraint targets sink $a$ & $^\dagger$Hard (trace) \\
\midrule
$B$-$B$ conflict & $B\,\pi{=}v_1$; $B\,\pi{=}v_2$; no revision segment & Soft \\
Bare-reassertion conflict & Bare $\pi{=}v_1$; bare $\pi{=}v_2$; no revision segment & Soft \\
Unjustified modal shift & $B{\to}K$ opp., $K{\to}B$ same/opp.; no derivation & Soft \\
\midrule
Unresolved doubt & $\delta(\pi) = 1$ at sink $a$ & Soft* \\
Unresolved unknowability & $\tau(\pi) = U_k$ at sink $a$ & Soft* \\
\midrule
Cross-path disagreement & Distinct terminal values on $\pi$ across paths & Soft (global) \\
\midrule
Redundant re-assertion & $K\,\pi \to K\,\pi$, no intervening content & Quality \\
Self-questioned $K$/$B$ & $K\,\pi \to K\,\mathord{?}\,\pi$ (analog for $B$) & Quality \\
Unverifiable derivation & Kleene $U_k$ derivation, committed $T/F$ target & Quality \\
Ambiguous negated connective & $\mathrm{NOT}$ folded into a connective & Quality \\
\bottomrule
\end{tabular}
\caption{Type-6 contradiction and quality-signal taxonomy. 
`Soft*' severity is elevated to `hard' when $\pi$ is in the direct implication chain to $a$. 
$^\dagger$ marks if the signal blocks the path from being considered coherent for aggregation.
}\label{tab:t6-contra}
\end{table}

\subsection{From Type-6 to Type-6-based Verification}\label{sec:t6-layers}

Type-6 is only a logic. 
Our CoT verifier works by (1) parsing the trace and applying Type-6 to its statements, and then (2) evaluating these statements through symbolic verification. 
This induces a natural hierarchical relationship between Type-6 and this verifier. 
For (1), applying Type-6 means that the labelling is \textit{implicit} in the CoT: a competent interpreter simply surfaces it. 
In (2), this evaluation necessarily induces a (multi di)graph: propositions become nodes, statements become edges, and structural edges (\textsf{gap}, \textsf{jump}, \textsf{loopback}, \textsf{meander}) follow from the sequential structure of the trace. 
Multiple edge-simple paths from source to sink are candidate reasoning routes. 
The rules by which symbolic verification occurs are Type-6 semantics over the sequential states $\Sigma$ from the graph. 
Comparing projected terminal states show two divergence modes: \emph{cross-path disagreement}, 
where two internally consistent paths yield opposite values on $\pi$ (the trace is structurally underdetermined), 
and \emph{under-resolution}, where some paths resolve $\pi$ and others leave it as $U_k$. 
The \emph{direct implication chain} to $a$ is the proposition set reaching $a$ through the transitive closure of constraints along a path. 
Residual signals on chain propositions are elevated to hard: an unresolved conclusion-supporting premise is not the same as an unresolved aside. 
This chain is computed per path, so a proposition may be on or off-chain based on the path.

\section{Methods and Experimental Setup}\label{sec:eval}

\subsection{Type-6-based Verification}\label{sec:verifier}

We now discuss our concrete verifier implementation. 
For flexibility in application, our main constraint is that a verifier should not be an author (i.e., fill in missing warrants, resolve underdetermined propositions beyond back-prop, or overrule the agent's committed values when they conflict with derived ones). 
The verifier is then an algorithm which takes in a raw CoT, annotates it in Type-6, and then produces reports from a sequential walk on a multi-digraph built from the annotation. 
It runs in quasilinear time when accounting for CoT sparsity (\appref{complexity}), and operates in five stages. 
For a worked example of Type-6, verification and a comparison with baseline outputs, see \appref{worked-example}. 

\textbf{Stages 1-3: Type-6 annotation.} 
The verifier (1) prompts an LLM to split and normalise statements, and assign them propositional variables: 
`I saw a dog but not a cat' becomes \texttt{p1:} \texttt{I} \texttt{saw} \texttt{a} \texttt{dog} and \texttt{AND} \texttt{NOT} \texttt{p2:} \texttt{I} \texttt{did} \texttt{not} \texttt{see} \texttt{a} \texttt{cat}. 
Then (2) function call maps recurring propositions to the same symbol, 
and (3) an LLM assigns per-statement values from $\{T, U_c, U_k, F\}$ and Type-6 operators to each line-statement using that map. 
It is also requested to revise any propositions missed. 
The LLM used was GPT-5 \citep{gpt5}.

\textbf{Stages 4-5: verification.} 
The verifier then builds a graph of the CoT applying the edge-kind rules of \secref{t6-layers}: each operator becomes an edge, each proposition a node, and consecutive statements yield \textsf{gap}, \textsf{jump}, \textsf{loopback}, or \textsf{meander} edges. 
Then, it enumerates all paths by computing all edge-simple paths from source (the user prompt) to sink (the answer), subject to a loopback budget. 
Node revisits are permitted only through \textsf{jump}, \textsf{loopback}, or binary-connective edges. 
From the graph and the sequential verifier state, the verifier produces three reports: 
\emph{coherence}, aggregating structural statistics across paths as distributions (min, max, mean per criteria); 
\emph{sequential verification}, logging contradictions and quality signals in order 
and \emph{per-path}, projecting sequential events onto individual paths and 
tabulating cross-path disagreements and under-resolutions. 

\subsection{Corpus creation and verifier review}\label{sec:data-models-baselines}

Our corpus is comprised of 996 CoT traces elicited by prompting two Qwen3 models (8B and 17B; \citealt{yang2025qwen3technicalreport}) on questions randomly drawn from three sources and two core domains (multiple choice and argumentation). 
For multiple-choice, we sampled 276 entries from CoT-Logic, a chain-of-thought benchmark;\footnote{\url{https://huggingface.co/datasets/isaiahbjork/cot-logic-reasoning}} 
and 370 from MMLU \citep{hendrycks2021measuring}. 
For argumentation we collected 311 arguments from DebateLab-KIT's arguments-and-debates
collection.\footnote{\url{https://huggingface.co/datasets/DebateLabKIT/arguments-and-debates}} In the argumentation subset, the models were instructed to respond to the prompt with a counterargument. 
In addition, we hand-crafted 39 question-answer pairs covering cases underrepresented in the sampled sources, mixed between argumentation and multiple choice. 
The median resulting trace length (as line-statements) is 66 for CoT-Logic, 54 for DebateLab, 44 for MMLU, and 52 for the handcrafted subset, for an average of 55.5. 
Multiple-choice traces probe deductive and eliminative reasoning under
formal constraints, while argumentation traces exercise informal
reasoning where enthymemes and rhetorical moves dominate. Both span the domains that Type-6 is designed for. 
Finally, we stripped the CoTs of model-specific tokens (\texttt{<think>}\ldots\texttt{</think>} and analogues).

To ensure reliability of Stages 1-3, three annotators trained on Type-6 logic reviewed the annotated traces in full over ten months. 
They evaluated proposition and operator-and-proposition assignment correctness. 
If the latter was incorrect, they were asked to provide a correction. 
Human annotator via Fleiss' $\kappa$ showed a raw pairwise agreement of $\kappa$=0.969 and 0.929 for the first and second criteria. 
The LLM did not annotate proposition correctness since it is not part of Type-6, but its accuracy on operator-and-proposition assignment was $97.2\%$, as per the corrections supplied by at least two annotators, split between $98.5\%$ (unverified) and $95.9\%$ (verified). 
Further details are in \secref{ethics-statement}.

\subsection{Measurement}

We compare against five methods, each covering at least one of the applications of Type-6 and its verifier. 
More details are in \appref{extended-methods}. 
All baselines only work with the natural-language trace, so the comparison is at the level of the raw CoT. 
This is in line with our verifier's intended use as \textit{post-hoc} evaluator. 
\textbf{ROSCOE} is a reference-free set of metrics. 
It measures faithfulness, informativeness, repetition, reasoning alignment, chain
self-consistency, and step contradiction. 
We aggregated it to a $[0,5]$ score with a hard-fail flag on the self-consistency and reasoning-alignment criteria. We use DeBERTa \citep{he2021deberta} as the NLI model. 
In \textbf{LINC} an LLM translates the trace into FOL, and formal verifier Z3 \citep{10.1007/978-3-540-78800-3_24} checks for premise consistency and entailment. 
The LLMs used are GPT-5, Gemma-4 \citep{gemma4}, GLM-4.7 \citep{glm47}, and Qwen-3.5. 
Verdicts are entailment, non-entailment, contradiction (treated as hard-fail), or unknown. We compute translation fidelity on every entry. 
\textbf{PRM}s assigns per-step correctness probabilities from \texttt{+}/\texttt{-} token logits, and hard-fails when mean or minimum step-probability falls below thresholds. A low-variance mid-range proxy flags likely out-of-distribution operation.
The models used are the ones provided by \citep{wang2024mathshepherd,dong2024tmlr-rlhf}, and \citep{zhang-etal-2025-lessons}. 
Finally, we also test \textbf{LLMs-as-judges} with a five-criteria rubric (contradiction, unsupported conclusion, modal mismatch, unresolved doubt, and an overall $[0,5]$ coherence rating). 
Hard-fails happen when any binary axis fires. 
The LLMs-as-judges used are GPT-5.6, GLM-4.7, Claude Opus-4.8 \citep{opus48}, and Llama-3.1 \citep{llama31}. 
We additionally consider a \textbf{trivial} baseline with four model-free signals: statement count, connective density, adjacent-statement repetition rate, and type-token ratio. 
Hard-fails on too-short or connective-free traces.

\section{Results}\label{sec:results}

\subsection{General Results}

We report two criteria as a normalisation factor for all six methods: 
(1) hard-fail agreement, where each method emits a binary flag per trace (accept/reject as incoherent) and measured with pairwise Cohen's $\kappa$; 
(2) scalar correlation, where each method emits a $[0,5]$ trace-level coherence score. 
For the latter our verifier used the coherence score from the report--that is, the direct verifier output--which is the most sensible reduction for this task. 
Pairwise hard-fail agreement and scalar correlations are in \tabref{hard-fail-agreement-correlation}. 
Breakdowns by generation model and task type for both metrics are in \appref{extended-results}. 
From the table it can be seen that pairwise hard-fail agreement is vanishingly small ($\kappa$ $\in$ $[-0.005, 0.15]$) among nearly all method pairs. 
The highest values (LLM-judge-Type-6, $\kappa = 0.15$; LLM-judge-PRM, $\kappa = 0.07$) indicate weak agreement. 
\textbf{There are no methods whose hard-fail decisions can be treated as
interchangeable}, but Type-6 is the most agreed-with. 
This holds when we take best-of or average scores. 
ROSCOE and Trivial rows show $\kappa \approx 0$, suggesting near-uniform emission of a single label. 
These verifiers \textbf{are labelling different subsets of the traces as failing}, and their disagreement is not due to the underlying model. 
The scalar correlations support this observation. 
Some method pairs weakly correlate: LLM-judge-PRM ($\rho = 0.36$ best-of; $0.13$ avg.),
LLM-judge-Trivial (r. $0.34$, $<0$), 
and PRM-Type-6 (r. $0.31$, $0.21$).
ROSCOE anti-correlates with all ($\rho$ $\in$ $[-0.40, -0.03]$), 
which, combined with its 100\% hard-fail rate, indicates it is doing something systematically different from the other methods. 
Overall, $\rho$, while non-trivial in a few pairs, \textbf{is much lower than expected} were these methods be measuring the same underlying quantity. 
No pair exceeds $\rho = 0.4$, and most sit below $0.2$.

\begin{table}[t]
\centering
\small
\setlength{\tabcolsep}{1.5pt}
\begin{tabular}{lrrrrrr|rrrrrr}
\toprule
 & LINC & LLM-J & PRM & ROSC. & Trivial & Type-6 & LINC & LLM-J & PRM & ROSC. & Trivial & Type-6 \\
LINC      & ---    & \cellcolor{BrickRed!20}$0.034$  & \cellcolor{BrickRed!20}$0.004$ & \cellcolor{BrickRed!20}$0.000$ & \cellcolor{BrickRed!20}-$0.005$& \cellcolor{BrickRed!20}$0.012$ & \cellcolor{BrickRed!20} $0.053$ &\cellcolor{BrickRed!20} $0.050$ &\cellcolor{BrickRed!20}-$0.043$ &\cellcolor{BrickRed!20} $0.024$ &\cellcolor{BrickRed!20} $0.005$ \\
LLM-J & \cellcolor{RoyalPurple!20}$0.004$ & ---     & \cellcolor{BrickRed!20}$0.067$ & \cellcolor{BrickRed!20}$0.000$ & \cellcolor{BrickRed!20}$0.002$ & \cellcolor{BrickRed!20}$0.151$ & \cellcolor{RoyalPurple!20}-$0.000$ & --- &\cellcolor{BrickRed!20} 0.363  &\cellcolor{BrickRed!20}  -$0.030$ &\cellcolor{BrickRed!20} $0.261$ &\cellcolor{BrickRed!20} $0.248$ \\
PRM       & \cellcolor{RoyalPurple!20}$0.001$ & \cellcolor{RoyalPurple!20}$0.054$ & ---     & \cellcolor{BrickRed!20}$0.000$ & \cellcolor{BrickRed!20}$0.000$ & \cellcolor{BrickRed!20}$0.074$ & \cellcolor{RoyalPurple!20}$-0.013$ & \cellcolor{RoyalPurple!20}$0.125$ & --- &\cellcolor{BrickRed!20} -$0.404$ &\cellcolor{BrickRed!20} -$0.093$ &\cellcolor{BrickRed!20} $0.312$ \\
ROSC.    & \cellcolor{RoyalPurple!20}$0.000$ & \cellcolor{RoyalPurple!20}$0.000$ & \cellcolor{RoyalPurple!20}$0.000$ & ---     & \cellcolor{BrickRed!20}$0.000$ & \cellcolor{BrickRed!20}$0.000$ & \cellcolor{RoyalPurple!20}-$0.054$ & \cellcolor{RoyalPurple!20}-$0.214$& \cellcolor{RoyalPurple!20}-$0.258$ & --- &\cellcolor{BrickRed!20} $0.068$  & \cellcolor{BrickRed!20}-$0.255$ \\
Trivial  &  \cellcolor{RoyalPurple!20}-$0.004$ & \cellcolor{RoyalPurple!20}$0.001$ & \cellcolor{RoyalPurple!20}$0.000$ & \cellcolor{RoyalPurple!20}$0.000$ & ---     & \cellcolor{BrickRed!20}$0.001$ & \cellcolor{RoyalPurple!20}$-0.024$ & \cellcolor{RoyalPurple!20}$0.343$ & \cellcolor{RoyalPurple!20}-$0.169$ & \cellcolor{RoyalPurple!20}$0.068$ & --- &\cellcolor{BrickRed!20} -$0.018$ \\
Type-6   &  \cellcolor{RoyalPurple!20}$0.005$  & \cellcolor{RoyalPurple!20}$0.092$ & \cellcolor{RoyalPurple!20}$0.057$ & \cellcolor{RoyalPurple!20}$0.000$ & \cellcolor{RoyalPurple!20}$0.001$ & --- & \cellcolor{RoyalPurple!20}$0.035$ & \cellcolor{RoyalPurple!20}$0.170$ & \cellcolor{RoyalPurple!20}$0.212$ &  \cellcolor{RoyalPurple!20}-$0.255$  & \cellcolor{RoyalPurple!20}-$0.018$ & --- \\
\bottomrule
\end{tabular}
\caption{Pairwise Cohen's $\kappa$ for hard-fail agreement (left) and Spearman $\rho$ between trace-level scalars (right), as average over per-method models (blue) and best-of (red). 
For Spearman $\rho$, all outputs project to $[0,5]$; 
Type-6' scalar is \texttt{coherence}\_\texttt{score}\_\texttt{graded}. 
LINC's $\rho$ is only for demonstration: our projection took verdict to either error, unknown, contradiction, non-entailment, or entailment, which is not fully comparable with the others. 
}
\label{tab:hard-fail-agreement-correlation}
\end{table}

\subsection{Ablation: Per-Method Measurements}\label{sec:per-baseline}

Our results require a per-baseline breakdown to understand \emph{how} each succeeds or fails. 
Further details are in \appref{per-method-tables}.
By construction, Trivial hard-fails almost no traces ($0.2\%$), but correlates non-trivially with LLM-judges ($\rho = 0.34$ with best-model; \tabref{hard-fail-agreement-correlation}), suggesting that \textbf{some of what LLM-judges term coherence is captured by trace surface features like length
and lexical diversity.} 
This does not necessarily mean a false-positive, however. 
Conversely, ROSCOE hard-fails on $100\%$ of the corpus and its scalar anti-correlates weakly-to-moderately with every other method ($\rho \in [-0.40, -0.03]$), 
despite its internal metrics spanning informative ranges. 
LINC's performance heavily depends on its translation model. 
Only GPT-5.6 produces informative verdicts ($14.7\%$ entailment, $65.3\%$ non-entailment, $4.0\%$ contradiction hard-fail), but does so from a very partial formalisation (coverage median $0.006$). 
Qwen-3.5-9B and GLM-4.7 emit $85.6\%$ and $72.8\%$ unknown, respectively, and error out otherwise. 
All PRMs exhibit collapse: Llama-3.1 hard-fails $68.4\%$; Math-Shepherd $96.2\%$; Qwen-2.5 $100\%$. 
They were trained on mathematical reasoning, but our corpus, like most CoT, is mostly natural-language argumentation. 
Zero low-confidence emissions from Llama and Qwen indicate that these PRMs are confidently mislabelling, rather than expressing uncertainty. 
LLM-judges grouped in two: capable judges (Claude Opus, GPT-5.6, Gemma-4B) hard-fail $61$-$68\%$ of traces with moderate coherence scores ($3.42$-$4.14$) and agree modestly with each other (within-judge $\kappa$ up to $0.38$). 
The rest hard-fail $98$-$100\%$ of traces, with pairwise $\kappa$ below $0.01$ against the capable judges. 
\textbf{LLM-as-judge coherence measurements, on this corpus, are dominated by the judge model rather than by the trace.} 
Through its findings, Type-6 shines as a structural verifier:
derived contradiction is the most common hard-fail ($409$ traces, mean $5.28$ events per triggered trace), followed by reasoning-avoidance ($548$) and unresolved doubt at sink ($495$). 
Direct-implication-chain statistics show median on-chain fraction of $0.025$ (IQR $[0.016, 0.074]$): \textbf{only $\sim$3\% of proposition-level statements in a typical trace directly bear on the derivation of the sink}; 
the remainder is context, exposition, or restatement. 
Type-6's alternative graded reductions correlate strongly ($\rho \geq 0.70$; \tabref{type6-reductions}), indicating this degree of freedom does not qualitatively change trace ordering.

\section{Discussion}\label{sec:discussion}

Our main results and ablation show that CoT verification is brittle beyond definitions or corpus-level agreement, and depends strongly on the model used. 
Some results can be explained easily: Trivial partially predicts LLM-judge scores through surface features, because these features matter for reasoning. 
ROSCOE's degenerate hard-fail score captures a real property (its metric distributions are non-degenerate), but unrelated to the other methods' measurements. 
This is likely due to out-of-distribution effects, which is what we argue also caused PRM's collapse. 
However, informal logic and argumentation \textit{is} the formalism behind CoT. 
From this angle, although the failures from PRM and LINC ($74.9\%$ of our corpus was marked as `unknown') are perhaps expected, NLI-based methods should be more successful. 
True enough, although ROSCOE was insufficient, LLM-judges fared better. 
Their main failure mode was due to model quality and low--but not zero--performance. 
Thus, \textbf{CoT verification methods are not measuring the same thing} in terms of hard-fail decisions or of scalar scores, as they have large, model-and-data-split invariant disagreements. 

This is why we framed Type-6 around structure. Type-6 is not as a better point on the same `coherence' axis as the baselines, but it operates in a different regime, where the verifier imposes as few assumptions as possible on the trace, while still relying on a formal framework. 
Still, \textbf{the methods are not interchangeable}, but are complementary. 
For example, it is clear that LLM-judges methods are reasonable defaults for rough, approximate signals, but ultimately statistical and untrustworthy. PRM and LINC are effective within their domains of application. 
Type-6 overarches--but does not substitute--them in the sense that provides auditability at the statement-level; i.e., it answers `\emph{where} did the reasoning go astray?' or `is the reasoning good enough?', rather than `did it fail?'.

\section{Conclusion}\label{sec:conclusion}

Type-6 is a discrete, trace-driven update-semantic logic that verifies coherence at the level of the artefact: it treats the agent as the authoritative narrator, records commitments and contradictions with
statement-level provenance, and encodes contemporary CoT idiosincrasies like meandering and loopbacks as part of its rules. 
We built a verifier based on Type-6 and compared it against five baselines (Trivial, ROSCOE, LINC, PRM, LLM-as-judge) on a $996$-trace corpus spanning multiple-choice and argumentation. 
Among these, no pair produces interchangeable failure decisions ($\kappa < 0.4$ across all pairs; $\rho < 0.4$ across all scalar pairs). 
Type-6 is the most-agreed-with method: it correlates with PRMs ($\rho = 0.31$) and LLM-judges ($\rho = 0.25$) more strongly than those methods correlate with each other, while producing this signal from a stable structural framework rather than from an LLM's judgement of the trace. 
It also showed the ability to surface structural pathologies that black-box scorers miss entirely, such as reasoning avoidance (present in 548 of 996 traces) and on-chain unresolved commitments--where only $\sim$3\% of a trace's propositions actually drive its conclusion. 
Although developed for single-agent contemporary CoT, Type-6's update semantics and per-path decomposition extend without modification to multi-agent and structured reasoning systems such as tree-of-thoughts \citep{yao2023tree}. 
Whether Type-6 offers the same performance in that setup than what we observed here; 
as well as a completeness theorem against a paraconsistent semantic target (e.g., \citealt{arieli1998reasoning}), are left for future work.

\section*{Ethics Statement}\label{sec:ethics-statement}

Type-6 verifies the internal coherence of a reasoning trace, not the truth of its conclusions or the intent of its author, and should not be deployed as a factuality checker or a certificate of correctness. 
Conflating a high coherence score with trustworthiness would misuse the tool in exactly the way this paper argues against. 
Human annotators were paid \$25-47 USD/hour depending on seniority. 
They worked on a schedule they controlled over the review period. 
The traces they reviewed were drawn from public benchmarks and hand-crafted items and contain no personal, private, or identifying data. 
No annotator was asked to review harmful or distressing content; the reasoning traces concern logic problems, academic multiple-choice questions, and philosophical argumentation. 
Further annotation details are in \appref{extended-methods}. 
DebateLab-KIT and MMLU are readily available under ODC-By 1.0 and MIT licenses. 

\section*{Reproducibility Statement}\label{sec:reproducibility-statement}

The verifier is deterministic: given a fixed annotated trace, it produces identical reports on every run, and its behaviour is auditable at statement-level. 
We provide formal rules, metatheory, and connective semantics in \appref{t6-axioms}. 
All models, versions, parameters, hardware, and the annotation protocol are in \appref{extended-methods}. 
Our logic specification, parser, visualization code, annotation rubric, and annotated dataset are publicly available.\footnote{Artefacts are at \url{https://github.com/adewynter/type6}}

\section*{AI Use Statement}\label{app:ai-disclosure}

Since LLMs are the core object of evaluation, traces and part of the verifier system rely on them. 
They are also part of the standard SOTA approaches for the field, and so they were used as evaluators in the LINC, PRM, and LLM-as-a-judge baselines. 
Details on this are documented in \appref{extended-methods}. 
An LLM (Claude Opus-4.8) was used to aid in the development of the verifier codebase, and to polish writing. 
Unit tests were hand designed to ensure correctness of the code, along with manual verification. 
Writing was reviewed by the authors to ensure conciseness, clarity, and appropriate style. 
A separate LLM (GPT-5.6 Sol) was used to aid in the development of proofs. 
The proofs were verified by hand by the authors.

\DeclareRobustCommand{\DE}[3]{#2}
\DeclareRobustCommand{\VAN}[3]{#3}

\bibliography{biblio}
\bibliographystyle{iclr2027_conference}

\appendix
\section{Limitations}\label{app:limitations}

Our corpus scope covers multiple-choice and argumentation, with some minor mathematical reasoning questions (word problems). 
While the patterns found could be arguably extensible due to Type-6's weak assumption set (e.g., no anthropomorphism), the baseline result interpretation could vary. 
For example, PRM is geared towards arithmetic, not informal, reasoning. 
On that same vein, the ROSCOE and PRM hard-fail thresholds must be tuned in advance, which means a more complete evaluation would sweep parameters. Do remark, however, that such computational effort defeats the purpose of general-purpose CoT verification. 

This verification in Type-6 depends on the LLM annotation pipeline (Stages 1-3). 
It could be argued that annotation errors propagate into the verifier's reports. 
However, the high performance of GPT-5 (above 95\% accuracy) on their extraction suggests that it is probable this error propagation is not of a concern. 
What we would argue is a limitation is the dependence on closed models for annotation. For this we have released the corpus, in the hopes that others create and calibrate more extensible solutions. In our work we made an effort to incorporate open and closed models alike for reproducibility purposes. 

Finally, our Type-6 characterisation lacks a completeness theorem, even though we established rule-level consistency and paraconsistency by construction. We leave this for future work.

\section{Two Worked Examples}\label{app:worked-example}

\tabref{worked-example-outputs} shows a sample Type-6 annotated trace.
The graph produced by the verifier is in \figref{worked-example}.
The trace commits to two propositions ($p_{12}, p_{13}$). 
Both carry modal force ($K p_{12}$, $B p_{13}$).
Because the commitments agree with the annotator's labels ($\tau = F$ in both cases, $\alpha \in \{T, F\}$), no modal mismatch and no $K$-$K$ contradiction fires; 
Type-6 records the commitments per the \emph{agent as authoritative narrator} policy and moves on. 
What Type-6 does flag is structural: the trace opens an $\mathrm{IF}$ antecedent at line 2 ($\mathrm{IF}\ p_9\ \mathrm{OR}\ p_{10}\ \mathrm{AND}\ p_{11}$), and interrupted by $R$ without ever being closed by a matching $\mathrm{THEN}$ (\emph{malformed implication}). 
Hence $p_9, p_{10}$ are $U_k$ and remain unresolved at the sink. 
They lie off the direct implication chain, so they are surfaced as (off-chain) \emph{unresolved-unknowability} residuals rather than elevated to hard. 
The other methods draw other conclusions: 
LLM-as-judge  flags an unsupported conclusion and modal mismatches on lines accepted by Type-6 as internally consistent; 
ROSCOE and PRM hard-fail by threshold (r. NLI self-consistency, $<0.5$; per-step score $<0.2$. 
LINC errors on translation without producing a verdict. 
Only Type-6 passes ($3.5$), penalising the malformed implication  but recognising the derivation as structurally sound.

In \figref{patterns} we show a more theoretical worked example, related to the remaining appendices. 
It has four canonical patterns from the edge-semantics rules, preserving the same colouring from the applied example from \figref{worked-example}. 
Solid black: \textsf{logical}.
Dashed gray: \textsf{gap}. 
Dotted purple: \textsf{jump}. 
Dashed orange: \textsf{loopback}. 
Dashed teal: \textsf{meander}. 
Blank circles are structural; coloured are propositions (green source, red sink, blue intermediate). 
In the worst case (\appref{comp-worst}) enumeration is combinatorial in parallel and revisit edges; 
the cap $K$ and the sampling fallback bound this. 
In the mostly-linear regime (\appref{mostly-linear}), describing realistic CoT graphs, the pipeline is polynomial in trace length and, under typical sparsity, effectively linear. 
The distributional reporting of coherence, verification, and per-path outputs (\secref{eval}), is a design commitment: coherence is not a scalar, and premature reduction discards the structure the graph exposes. 

\begin{table}[t]
\centering
\small
\begin{tabular}{lll}
$T$  & \texttt{q}       & Explain how it relates to other thought experiments in philosophy of mind. \\
$U_k$ & \texttt{IF p9}   & If a nation of a billion people can implement (...) a mind, \\
$U_k$ & \texttt{OR p10}  & or if the same profile can be implemented in silicon, \\
$T$  & \texttt{AND p11} & and if the phenomenal properties supervene only on this profile, \\
$U_k$ & \texttt{N}       & Hmm. \\
$F$  & \texttt{K p12}   & Wait, I remember that functionalism entails multiple realizability trivially. \\
$F$  & \texttt{B p13}   & Maybe the Chinese Nation is a specific case of that entailment. \\
$U_k$ & \texttt{R}       & Let me try to structure this. \\
$U_k$ & \texttt{N}       & First, restate the argument. \\
$T$  & \texttt{p14}     & The argument shows that a system with (...) is a mind, \\
$T$  & \texttt{AND p15} & and yet nothing about the Chinese Nation intuitively resembles one. \\
$T$  & \texttt{THEN a}  & So, functional organization alone is not sufficient for (...) consciousness. \\ \midrule
\textbf{Method} & \textbf{Hard fail} & \textbf{Salient signal} \\
\midrule
\textit{Trivial} & yes & length=12, conn density=0.08 \\
\textit{ROSCOE} & yes & self-consistency 0.06 < 0.5; mean RA 0.15 < 0.4 \\
\textit{LINC} & no & errored: no conclusion found. \\
\textit{PRM} & yes & min step score 0.13 < 0.2 \\
\textit{LLM-judge} & yes & unsupported conclusion; modal mismatch; unresolved doubt \\
\textit{Type-6} & no & unresolved doubt; malformed implication \\
\bottomrule
\end{tabular}
\caption{Method-by-method output on a worked example asking how \cite{block1978troubles}'s Chinese Nation thought experiment relates to other arguments in philosophy. 
At the example's leftmost column we note the truth values of each line. 
These are not used by any of the baselines except Type-6' verifier. 
LLM-as-a-judge (GPT-5.6) provides a detailed rationale matching Type-6. 
Remark how each method save LINC (Opus-4.8) and Type-6 hard-fails due to low scoring (e.g. PRM; Math-shepherd). 
LINC does not hard fail because it fails to find a conclusion and hence it cannot draw a conclusion. 
Type-6 passes because it tolerates hedged reasoning ($p_{13}$) and the structural issues are off-chain.}
\label{tab:worked-example-outputs}
\end{table}

\begin{figure}
    \centering
    \includegraphics[width=\linewidth]{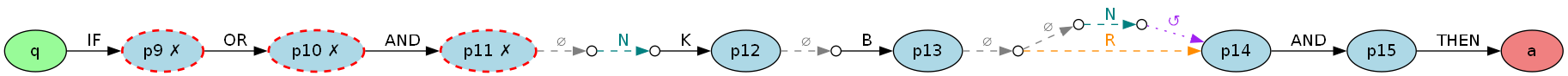}
    \caption{Sample multidigraph generated by the Type-6 verifier on \tabref{worked-example-outputs}. 
The code provided renders source and sink as green and red, respectively, and propositions in blue. 
Failures are bordered with dashed red and are marked with $\mathcal{X}$. 
Orange borders mark a soft contradiction. Jumps are denoted by $\varnothing$. 
}
    \label{fig:worked-example}
\end{figure}

\begin{figure}[t]
\centering
\begin{tikzpicture}[
  >=stealth, node distance=8mm,
  prop/.style={circle, draw, fill=blue!15, minimum size=6mm, inner sep=0pt, font=\small},
  src/.style ={circle, draw, fill=green!25, minimum size=6mm, inner sep=0pt, font=\small},
  snk/.style ={circle, draw, fill=red!25, minimum size=6mm, inner sep=0pt, font=\small},
  blank/.style={circle, draw, fill=white, minimum size=2mm, inner sep=0pt},
  logical/.style ={->, thick},
  gap/.style     ={->, dashed, gray},
  jump/.style    ={->, dotted, purple, thick},
  loop/.style    ={->, dashed, orange!80!black},
  mean/.style    ={->, dashed, teal},
  edgelabel/.style={font=\scriptsize, midway}
]
\node[prop] (a0) {$p_{13}$};
\node[blank, right=of a0] (a1) {};
\node[blank, above right=6mm and 10mm of a1] (a2) {};
\node[prop, right=14mm of a1] (a3) {$p_1$};
\draw[gap]    (a0) -- node[edgelabel, above]{gap} (a1);
\draw[loop]   (a1) -- node[edgelabel, below]{R}   (a3);
\draw[jump]   (a1) -- node[edgelabel, above left]{$\circlearrowleft$} (a2);
\draw[logical](a2) -- node[edgelabel, above right]{K} (a3);
\node[below=8mm of a1, font=\small] {(a) single $R$ with parallel op-chain};
\node[prop, below=22mm of a0] (b0) {$p_{13}$};
\node[blank, right=of b0] (b1) {};
\node[blank, right=of b1] (b2) {};
\node[prop, right=of b2] (b3) {$p_1$};
\draw[gap] (b0) -- node[edgelabel, above]{gap} (b1);
\draw[loop](b1) -- node[edgelabel, above]{R}   (b2);
\draw[loop](b2) -- node[edgelabel, above]{R}   (b3);
\node[below=6mm of b2, font=\small] {(b) chained $R$: one gap, serial composition};
\node[prop, below=22mm of b0] (c0) {$q$};
\node[blank, right=of c0] (c1) {};
\node[blank, right=of c1] (c2) {};
\node[prop, right=of c2] (c3) {$p_1$};
\draw[gap]    (c0) -- node[edgelabel, above]{gap} (c1);
\draw[mean]   (c1) -- node[edgelabel, above]{N}   (c2);
\draw[logical](c2) -- node[edgelabel, above]{K}   (c3);
\node[below=6mm of c2, font=\small] {(c) $N$ pivot: gap in, no gap out};
\node[prop, below=22mm of c0] (d0) {$p_1$};
\node[blank, right=of d0] (d1) {};
\node[prop, right=of d1] (d2) {$p_{14}$};
\node[blank, right=of d2] (d3) {};
\node[prop, right=of d3] (d1r) {$p_1$};
\draw[gap]    (d0) -- node[edgelabel, above]{gap} (d1);
\draw[logical](d1) -- node[edgelabel, above]{K}   (d2);
\draw[jump]   (d2) -- node[edgelabel, above]{$\circlearrowleft$} (d3);
\draw[logical](d3) -- node[edgelabel, above]{K}   (d1r);
\node[below=6mm of d2, font=\small] {(d) gap introduces new material; jump returns};
\end{tikzpicture}
\caption{Canonical edge-semantics patterns produced by $\textsc{BuildGraph}$ (\algref{build}).}
\label{fig:patterns}
\end{figure}

\section{Extended Methods}\label{app:extended-methods}

\subsection{Models, Hardware, and Parameters}\label{app:extended-models}

The models used in this work are GPT-5 (reasoning-large), GPT-5.6 (reasoning-sol), GPT-4.1 (gpt-41-longco-2025-04-14) Claude Opus-4.8, Gemma 4 e4B, GLM-4.7-Flash, and Llama 3.1. 
We used Qwen3 (8B, 17B) for trace generation and Qwen3.5 9B for LINC and LLM-as-a-judge. 
For ROSCOE we used DeBERTA-large. 
The GPT-class models as well as Claude were called through the Azure OpenAI API. 
The rest were called in a consumer-grade computer with 24 CPU cores and 60 GPU cores. 

For trace generation and annotation, the Qwen3 models were called with maximum tokens set at 12,000 and temperature $0.8$. 
Outside of this step, all models had their temperature set to zero whenever possible. 
Proposition extraction, truth-value annotation, and operator assignment were ran primarily on GPT-5 (maximum tokens: 15,000), and GPT-4.1 (maximum tokens=8000) was configured as a fallback for malformed structured outputs. 
It was invoked on two traces, and so these were discarded due to low annotation quality. 
The verify-and-correct pass ran also on GPT-5 with the same parameters. 
Prompts for all passes are in the repository. 

LLMs-as-judges used 1,024 maximum return tokens. 
LINC models used 4,092 tokens; the version of z3 used is 0.2.0. 
For ROSCOE, the self-consistency threshold was default at 0.5 and the reasoning alignment threshold at 0.4. Calls were with a batch size of 32. 
PRM models were set at a 2,048 maximum return tokens. 
All token parameters were tuned in a 5\% subset of the corpus before full usage. High failure rates (e.g., Llama on LINC) meant the model was discarded.

\section{Type-6: Axiomatic Presentation and Metatheory}\label{app:t6-axioms}

This appendix presents Type-6 as an axiom system with inference rules,
gives a worked example of nested doubt scoping, and proves its basic
metatheoretic properties. Throughout, $\Sigma$ ranges over states,
$\mathcal{H}$ over chain states, $\mathcal{K}$ over persistent
constraint sets, $\pi$ over propositions, and $v$ over $\{T, F\}$ unless
noted (all as defined in \secref{type6}).

\subsection{Judgements}\label{app:t6-judgements}

Let $c \in \mathcal{K}$ be $(\pi, \epsilon_t) \equiv \phi(p_1, \ldots, p_k)$ (definitional identity) or $\phi(p_1, \ldots, p_k) \to (\pi, \epsilon_t)$ (implication), where $\phi$ is a Boolean combination and $\epsilon_t \in \{+, -\}$ the target polarity. 
Then we say that the constraint set \textbf{forces} a value $v^*$, $\mathcal{K} \models_\Sigma \pi = v^\ast$, iff evaluating $\phi$ under $\Sigma$ via the tables of \appref{t6-connectives} yields a resolved $w \in \{T, U_k, F\}$  (not abstention $\ast$), with $v^\ast = w$ if $\epsilon_t = +$ and $v^\ast = \text{flip}(w)$ if $\epsilon_t = -$. 

Here, $\text{flip}(T) = F$, $\text{flip}(F) = T$, $\text{flip}(U_c) = U_c$, and $\text{flip}(U_k) = U_k$. 
For implications, forcing fires only when $\phi$ evaluates to $T$; when $\phi = F$ the implication is vacuously satisfied and forcing does not fire.

Hence, Type-6 has three judgement forms:
\begin{itemize}
\item $\Sigma \vdash \pi : (\tau, \sigma, \delta)$ -- in state $\Sigma$, proposition $\pi$ carries the indicated triple.
\item $(\Sigma, \mathcal{H}, \mathcal{K}) \xrightarrow{s} (\Sigma', \mathcal{H}', \mathcal{K}')$ -- statement $s$ transforms the configuration.
\item $\mathcal{K} \models_\Sigma \pi = v^\ast$: the constraint set forces value $v^\ast$ on $\pi$ under $\Sigma$.
\end{itemize}

\subsection{Inference Rules}\label{app:t6-rules}

We assume a global annotation function $\alpha$ assigning each statement-proposition pair a value in $\mathcal{V}$. 
Rules operate on $(\Sigma, \mathcal{H}, \mathcal{K})$; untouched components are elided from conclusions. 
Chain state $\mathcal{H} = (\chi, \phi)$ carries an optional open chain $\chi = ((p_1, \beta_1, \epsilon_1), \ldots)$ of operand triples (proposition, joining connective $\beta_i \in \{\mathrm{seed}, \mathrm{AND}, \mathrm{OR}\}$, and polarity $\epsilon_i \in \{+, -\}$) 
and a pending-$\mathrm{IF}$ flag $\phi \in \{\top, \bot\}$. 
We write $\chi \cdot (p, \beta, \epsilon)$ for extension and $\varnothing$ for the empty chain.

We now define two crucial operations in Type-6 and its verification. 

\begin{definition}[close]\label{def:t6-close}
Given $(\chi, \phi)$ and a target $\pi$ with polarity $\epsilon_t$, let $\Psi(\chi)$ be the Boolean combination obtained by evaluating operands left-to-right under AND-binds-tighter-than-OR, with $p_i$ as $p_i$ if $\epsilon_i = +$ else $\neg p_i$, joined by $\beta_i$. 
Then $\mathrm{close}(\phi, \chi, \pi, \epsilon_t)$ returns $(\pi, \epsilon_t) \equiv \Psi(\chi)$ if $\phi = \bot$, or $\Psi(\chi) \to (\pi, \epsilon_t)$ if $\phi = \top$.
\end{definition}

\begin{definition}[terminate]
\label{def:t6-terminate}
$\mathrm{terminate}(\phi, \chi, \mathcal{K})$ emits a \emph{malformed-implication} event and returns $\mathcal{K}$ if $\phi =\top$ (an unmatched $\mathrm{IF}$ is discarded); otherwise returns $\mathcal{K}$ unchanged.
\end{definition}

\textbf{Totality of state}. 
Note that $\Sigma_0$ assigns $(U_k, \bot, 0)$ to every proposition in the trace, so every antecedent $\Sigma(\pi) = (\tau_0, \sigma_0, \delta_0)$ below matches a well-defined triple; no rule needs a fresh-proposition case.

\subsubsection*{Modal rules}

\begin{figure}[H]
\centering
\small
\[
\begin{array}{c}
\infer[\text{(K-Commit)}]
  {(\Sigma, \mathcal{H}, \mathcal{K}) \xrightarrow{K\,\pi}
   (\Sigma[\pi \mapsto (\alpha, K, 0)], \mathcal{H}, \mathcal{K})}
  {\alpha \in \{T, F\}}
\\[2ex]
\infer[\text{(B-Commit)}]
  {(\Sigma, \mathcal{H}, \mathcal{K}) \xrightarrow{B\,\pi}
   (\Sigma[\pi \mapsto (\alpha, B, \delta_0)], \mathcal{H}, \mathcal{K})}
  {\Sigma(\pi) = (\tau_0, \sigma_0, \delta_0)}
\end{array}
\]
\caption{Modal rules. K-Commit clears $\delta$; B-Commit preserves it.
Both leave $\mathcal{H}, \mathcal{K}$ unchanged.}
\label{fig:t6-rules-modal}
\end{figure}

\subsubsection*{Assertion and structural rules}

\begin{figure}[H]
\centering
\small
\[
\begin{array}{c}
\infer[\text{(Bare-Assert)}]
  {(\Sigma, \mathcal{H}, \mathcal{K}) \xrightarrow{\pi}
   (\Sigma[\pi \mapsto (\alpha, \sigma_0, \delta_0)], \mathcal{H}, \mathcal{K})}
  {\Sigma(\pi) = (\tau_0, \sigma_0, \delta_0)}
\\[2ex]
\infer[\text{(Doubt)}]
  {(\Sigma, \mathcal{H}, \mathcal{K}) \xrightarrow{\mathord{?}\,\pi}
   (\Sigma[\pi \mapsto (\tau_0, \sigma_0, 1)], \mathcal{H}, \mathcal{K})}
  {\Sigma(\pi) = (\tau_0, \sigma_0, \delta_0)}
\\[2ex]
\infer[\text{(Loopback)}]
  {(\Sigma, (\chi, \phi), \mathcal{K}) \xrightarrow{R}
   (\Sigma, (\varnothing, \bot), \mathrm{terminate}(\phi, \chi, \mathcal{K}))}{}
\\[2ex]
\infer[\text{(Pivot)}]
  {(\Sigma, \mathcal{H}, \mathcal{K}) \xrightarrow{N}
   (\Sigma, \mathcal{H}, \mathcal{K})}{}
\end{array}
\]
\caption{Assertion and structural rules. Bare-Assert fires when a
statement asserts $\pi$ with no modal or connective operator in scope;
Doubt when $\mathord{?}$ is outermost. Pivot ($N$) preserves the
configuration. Loopback ($R$) preserves $\Sigma$ but resets the chain
and terminates any open $\mathrm{IF}$.}
\label{fig:t6-rules-assertion}
\end{figure}

Remark that \textbf{loopback preserves} $\Sigma$. Its role in revision is captured by the presence of an $R$-edge in the sub-trace, not by any state change; whether a subsequent commitment is a licensed revision is determined by the revision-segment predicate over the trace.

\subsubsection*{Chain-accumulator rules}

\begin{figure}[H]
\centering
\small
\[
\begin{array}{c}
\infer[\text{(Bare-Seed)}]
  {(\Sigma, (\varnothing, \bot), \mathcal{K}) \xrightarrow{\pi}
   (\Sigma', ((\pi, \mathrm{seed}, +), \bot), \mathcal{K})}
  {(\Sigma, \mathcal{H}, \mathcal{K}) \xrightarrow{\pi}
   (\Sigma', \mathcal{H}, \mathcal{K})}
\\[2ex]
\infer[\text{(And-Extend)}]
  {(\Sigma, (\chi, \phi), \mathcal{K}) \xrightarrow{\mathrm{AND}\,\pi}
   (\Sigma', (\chi \cdot (\pi, \mathrm{AND}, +), \phi), \mathcal{K})}
  {(\Sigma, \mathcal{H}, \mathcal{K}) \xrightarrow{\pi}
   (\Sigma', \mathcal{H}, \mathcal{K})}
\\[2ex]
\infer[\text{(Or-Extend)}]
  {(\Sigma, (\chi, \phi), \mathcal{K}) \xrightarrow{\mathrm{OR}\,\pi}
   (\Sigma', (\chi \cdot (\pi, \mathrm{OR}, +), \phi), \mathcal{K})}
  {(\Sigma, \mathcal{H}, \mathcal{K}) \xrightarrow{\pi}
   (\Sigma', \mathcal{H}, \mathcal{K})}
\\[2ex]
\infer[\text{(If-Open)}]
  {(\Sigma, (\chi, \phi), \mathcal{K}) \xrightarrow{\mathrm{IF}\,\pi}
   (\Sigma', ((\pi, \mathrm{seed}, +), \top), \mathcal{K}')}
  {(\Sigma, \mathcal{H}, \mathcal{K}) \xrightarrow{\pi}
   (\Sigma', \mathcal{H}, \mathcal{K})
   & \mathcal{K}' = \mathrm{terminate}(\phi, \chi, \mathcal{K})}
\\[2ex]
\infer[\text{(Then-Close)}]
  {(\Sigma, (\chi, \phi), \mathcal{K}) \xrightarrow{\mathrm{THEN}\,\pi}
   (\Sigma', (\varnothing, \bot), \mathcal{K} \cup \{\Phi\})}
  {(\Sigma, \mathcal{H}, \mathcal{K}) \xrightarrow{\pi}
   (\Sigma', \mathcal{H}, \mathcal{K})
   & \Phi = \mathrm{close}(\phi, \chi, \pi, +)}
\end{array}
\]
\caption{Chain-accumulator rules. Each carries a subgoal firing Bare-Assert on the operand or target (\figref{t6-rules-assertion}) and separately updates $\mathcal{H}$. 
If-Open terminates any prior open $\mathrm{IF}$ as malformed; 
Then-Close constructs $\Phi$ via $\mathrm{close}$ (\defref{t6-close}) and resets the chain.}\label{fig:t6-rules-chain}
\end{figure}

\subsubsection*{Modifier rules}

\begin{figure}[H]
\centering
\small
\[
\begin{array}{c}
\infer[\text{(NOT-Modal)}]
  {(\Sigma, \mathcal{H}, \mathcal{K}) \xrightarrow{\mathrm{NOT}\,\omega\,\pi}
   (\Sigma[\pi \mapsto (\bar\alpha, \sigma_1, \delta_1)], \mathcal{H}, \mathcal{K})}
  {\omega \in \{K, B\}
   & (\Sigma, \mathcal{H}, \mathcal{K}) \xrightarrow{\omega\,\pi}
     (\Sigma[\pi \mapsto (\alpha, \sigma_1, \delta_1)], \mathcal{H}, \mathcal{K})
   & \bar\alpha = \mathrm{flip}(\alpha)}
\\[2ex]
\infer[\text{(NOT-Operand)}]
  {(\Sigma, (\chi, \phi), \mathcal{K}) \xrightarrow{\beta\,\mathrm{NOT}\,\pi}
   (\Sigma', (\chi \cdot (\pi, \beta, -), \phi), \mathcal{K})}
  {\beta \in \{\mathrm{AND}, \mathrm{OR}\}
   & (\Sigma, \mathcal{H}, \mathcal{K}) \xrightarrow{\pi}
     (\Sigma', \mathcal{H}, \mathcal{K})}
\\[2ex]
\infer[\text{(NOT-Target)}]
  {(\Sigma, (\chi, \phi), \mathcal{K})
     \xrightarrow{\mathrm{THEN}\,\mathrm{NOT}\,\pi}
   (\Sigma', (\varnothing, \bot), \mathcal{K} \cup \{\Phi^-\})}
  {(\Sigma, \mathcal{H}, \mathcal{K}) \xrightarrow{\pi}
     (\Sigma', \mathcal{H}, \mathcal{K})
   & \Phi^- = \mathrm{close}(\phi, \chi, \pi, -)}
\end{array}
\]
\caption{Modifier rules. NOT-Modal composes with an inner K/B transition and flips the annotation ($\bar\alpha = \mathrm{flip}(\alpha)$). 
NOT-Operand records negated operand polarity;  
NOT-Target closes with a negated target. 
NOT-Operand and NOT-Target additionally emit an \emph{ambiguous-negated-connective} quality signal
(\secref{type6}).}
\label{fig:t6-rules-not}
\end{figure}

\subsubsection*{Back-propagation rules}

Back-propagation runs on $\mathcal{K}$ after any transition. 
Each $c \in \mathcal{K}$ is evaluated against $\Sigma$ under forcing; 
at most one rule applies per constraint per step.

\begin{figure}[H]
\centering
\small
\[
\begin{array}{c}
\infer[\text{(Backprop-Pin)}]
  {(\Sigma, \mathcal{H}, \mathcal{K}) \xrightarrow{\mathrm{backprop}}
   (\Sigma[\pi \mapsto (v^\ast, \sigma_0, \delta_0)], \mathcal{H}, \mathcal{K})}
  {\mathcal{K} \models_\Sigma \pi = v^\ast
   & \Sigma(\pi) = (U_k, \sigma_0, \delta_0)
   & v^\ast \in \{T, F\}}
\\[2ex]
\infer[\text{(Backprop-Conflict)}]
  {(\Sigma, \mathcal{H}, \mathcal{K}) \xrightarrow{\mathrm{backprop}}
   (\Sigma, \mathcal{H}, \mathcal{K})}
  {\mathcal{K} \models_\Sigma \pi = v^\ast
   & \Sigma(\pi) = (v, \sigma_0, \delta_0)
   & \{v, v^\ast\} = \{T, F\}}
\\[2ex]
\infer[\text{(Backprop-Unverifiable)}]
  {(\Sigma, \mathcal{H}, \mathcal{K}) \xrightarrow{\mathrm{backprop}}
   (\Sigma, \mathcal{H}, \mathcal{K})}
  {\mathcal{K} \models_\Sigma \pi = U_k
   & \Sigma(\pi) = (v, \sigma_0, \delta_0)
   & v \in \{T, F\}}
\end{array}
\]
\caption{Back-propagation rules. 
Pin promotes an $U_k$ target when its constraint forces a resolved value; 
Conflict retains the committed value under conflict; 
Unverifiable applies when the target is committed and the derivation is inconclusive.}
\label{fig:t6-rules-backprop}
\end{figure}

\textbf{Deduplication policy.} 
Note that Backprop-Conflict and Backprop-Unverifiable preserve $(\Sigma,
\mathcal{H}, \mathcal{K})$ and additionally emit an event (a \emph{derived-contradiction} and an \emph{unverifiable-derivation}, respectively) governed by the deduplication policy: 
each $c \in \mathcal{K}$ carries a \emph{last-logged-status} field (initialized $\bot$). Conflict and Unverifiable emit only on \emph{transition} of $c$'s status under $\Sigma$: if $c$'s outcome under $\Sigma_i$ equals its outcome under $\Sigma_{i-1}$, no event is emitted at step $i$.

\subsection{Nested Operator Scoping}\label{app:t6-nested}

Operator stacks are read left-to-right, each operator scoping over the stack to its right. 
For example, $\mathord{?}\,K\,\pi$ (`do I know $\pi$?'). 
Applying $K\,\pi$ (innermost) first, then $\mathord{?}$:
\begin{align*}
\Sigma
  &\xrightarrow{K\,\pi} \Sigma[\pi \mapsto (\alpha, K, 0)]
  \xrightarrow{\mathord{?}} \Sigma[\pi \mapsto (\alpha, K, 1)].
\end{align*}
The final state is $(\alpha, K, 1)$: 
knowledge committed on $\pi$, but the commitment is under doubt. 
On the other hand, $K\,\mathord{?}\,\pi$ (`I know I doubt $\pi$') applies $\mathord{?}\,\pi$ first, then $K$--whose target is the doubt-state, not $\pi$'s value. 
So $K$ does not update $\sigma(\pi)$; only the inner $\mathord{?}$ sets the flag:
\begin{align*}
\Sigma
  &\xrightarrow{\mathord{?}\,\pi} \Sigma[\pi \mapsto (\tau_0, \sigma_0, 1)]
  \xrightarrow{K} \Sigma[\pi \mapsto (\tau_0, \sigma_0, 1)].
\end{align*}
Hence the final state is $(\tau_0, \sigma_0, 1)$: the knowledge commitment attaches to the doubt, registered via $\delta = 1$; and the value on $\pi$ is untouched.

\subsection{Metatheoretic Properties}\label{app:t6-meta}

\begin{proposition}[Determinism]
\label{prop:t6-det}
For any configuration and statement $s$, at most one rule of \figstworef{t6-rules-modal}{t6-rules-not} applies, and the resulting configuration is unique. 
Back-propagation (\figref{t6-rules-backprop}) then fires deterministically: 
at most one of Pin, Conflict, Unverifiable applies per constraint per step.
\end{proposition}

\begin{proof}
The rules partition by the outermost operator of $s$: K-Commit iff $s =K\,\pi$; B-Commit iff $s = B\,\pi$; Bare-Assert iff $s$ is fully bare; 
Doubt iff $s = \mathord{?}\,\pi$ with no modal to its right; 
Pivot iff $s = N$; 
Loopback iff $s = R$; 
the NOT-rules iff the outermost operator is $\mathrm{NOT}$ in the corresponding position; 
the chain rules on connective-led statements. 
These are mutually exclusive by the grammar of $\mathcal{L}_{\text{Type-6}}$. 
Each yields an explicit functional update of $\Sigma, \mathcal{H}, \mathcal{K}$; 
NOT-rules recurse on the inner operator, unique by induction on operator depth. 
For back-propagation, the three rules partition on disjoint committed-value conditions:
$\tau(\pi) = U_k$ (Pin), $\tau(\pi) \in \{T,F\}$ with $v^\ast \in \{T,F\}, v^\ast \neq \tau(\pi)$ (Conflict), $\tau(\pi) \in \{T,F\}$ with $v^\ast = U_k$ (Unverifiable). 
Hence at most one applies per constraint per step.
\end{proof}

\begin{proposition}[Monotonicity of resolution]
\label{prop:t6-mono}
For any trace and $\pi$, the resolvedness of $\tau(\pi)$ is nondecreasing in $\sqsubseteq$ across $\Sigma_0, \Sigma_1, \dots$ \emph{modulo licensed revision}: 
absent a K-revision or bare-revision at step $i+1$, $\tau_i(\pi) \sqsubseteq \tau_{i+1}(\pi)$.
\end{proposition}

\begin{proof}
By cases. Modal/bare commits on $\pi' \neq \pi$, Doubt, Pivot, Loopback, and N preserve $\tau(\pi)$. Chain rules carry a Bare-Assert subgoal, so reduce to that case. 
NOT-Modal composes with an inner modal rule; 
polarity flip does not change $\tau$'s position in $\sqsubseteq$ (both $T, F$ maximal). 
Backprop-Pin promotes $U_k \sqsubseteq v^\ast$ (a strict promotion); 
Backprop-Conflict and Backprop-Unverifiable retain the committed value.
\end{proof}

\begin{proposition}[Local consistency]
\label{prop:t6-consist}
No single rule application of \figstworef{t6-rules-modal}{t6-rules-backprop} introduces a
contradiction in the sense of \secref{type6}, unless the contradiction is explicitly detected and logged by the verifier.
\end{proposition}

\begin{proof}
By case analysis:

\emph{K-Commit}: the side condition $\alpha \in \{T, F\}$ excludes $U_c, U_k$; 
those cases fire the logged modal-mismatch predicate rather than proceeding silently. 
A prior differing $K$-value with no revision segment fires the K-K contradiction; 
otherwise it is a licensed K-revision and nothing is logged.

\emph{B-Commit}: permitted on $U_c$ (logged as modal mismatch) and $U_k$ (permitted). 
A prior differing $B$-value with no revision segment fires a logged B-B conflict; 
unjustified modal shifts ($B \to K$, $K \to B$ without an intervening derivation) are logged as soft.

\emph{Bare-Assert}: permitted on all four values. 
A prior differing bare value with no revision segment fires the logged bare-reassertion conflict; 
otherwise a licensed bare-revision, nothing logged.

\emph{Doubt, Pivot}: do not modify $\Sigma, \mathcal{H}, \mathcal{K}$ and cannot introduce K-K, B-B, or modal-mismatch contradictions. 

\emph{Loopback}: does not modify $\Sigma$; 
resets $\mathcal{H}$ and may fire the malformed-implication event via $\mathrm{terminate}$, logged as a diagnostic, not a contradiction.

\emph{Chain rules}: f
ire their Bare-Assert subgoal (consistency reduces to that case); 
If-Open with a nonempty prior $\mathrm{IF}$ fires the malformed-implication event, logged, without silently corrupting $\mathcal{K}$.

\emph{NOT-Modal}: reduces to the inner rule (consistency by induction). 
NOT-Operand and NOT-Target additionally fire the ambiguous-negated-connective quality signal (non-blocking).

\emph{Backprop}: 
Pin promotes $U_k$ under a forcing constraint, no contradiction; any later inconsistency with a subsequent commit is detected at that later step. 
Conflict logs a derived-contradiction and retains the committed value; 
Unverifiable logs its event without modifying $\Sigma$. 
All three track status transitions per the deduplication policy.

In every case the rule leaves $\Sigma$ consistent, or the inconsistency is detected and logged before the update completes. This concludes the proof.
\end{proof}

\begin{remark}[Completeness]
\label{rem:t6-complete}
Type-6 is deliberately \emph{not} complete with respect to a classical propositional or modal semantics: 
the four-valued truth space and the tolerance of enthymemes preclude a completeness theorem of the DEL variety. 
Establishing completeness would require fixing a target semantics (e.g., a paraconsistent bilattice in the spirit of \citealt{belnap1977useful}). 
We leave that for future work.
\end{remark}

\begin{remark}[Paraconsistency]
\label{rem:t6-para}
A K-K contradiction on $\pi$ does not license arbitrary further conclusions: 
contradictions are logged and localized to $\pi$, and inference on other propositions is unaffected. 
\textbf{Type-6 is thus paraconsistent by construction}. 
This is appropriate for CoT traces, which are frequently locally inconsistent without being globally uninformative.
\end{remark}

\subsection{Semantics of Logical Connectives}\label{app:t6-connectives}

In this section we discuss the truth tables used by back-propagation over $\mathrm{THEN}$-led constraints. 
Operands draw from $\mathcal{V} = \{T, U_c, U_k, F\}$; $U_c$ propagates as a computational \emph{abstention} ($\ast$): 
any $U_c$ operand makes the enclosing derivation abstain rather than propagate. 
Negation on the underdetermined values is the identity, $\neg U_c = U_c$, $\neg U_k = U_k$, consistent with the no-op polarity flip in NOT-Operand and NOT-Target.

\begin{center}
\small
\begin{tabular}{c|cccc}
$\land$ & $T$   & $U_k$ & $U_c$ & $F$ \\
\hline
$T$     & $T$   & $U_k$ & $\ast$ & $F$ \\
$U_k$   & $U_k$ & $U_k$ & $\ast$ & $F$ \\
$U_c$   & $\ast$ & $\ast$ & $\ast$ & $F$ \\
$F$     & $F$   & $F$   & $F$   & $F$ \\
\end{tabular}
\qquad
\begin{tabular}{c|cccc}
$\lor$  & $T$ & $U_k$ & $U_c$ & $F$   \\
\hline
$T$     & $T$ & $T$   & $T$   & $T$   \\
$U_k$   & $T$ & $U_k$ & $\ast$ & $U_k$ \\
$U_c$   & $T$ & $\ast$ & $\ast$ & $\ast$ \\
$F$     & $T$ & $U_k$ & $\ast$ & $F$   \\
\end{tabular}
\qquad
\begin{tabular}{c|cccc}
$\to$   & $T$ & $U_k$ & $U_c$ & $F$   \\
\hline
$T$     & $T$ & $U_k$ & $\ast$ & $F$   \\
$U_k$   & $T$ & $U_k$ & $\ast$ & $U_k$ \\
$U_c$   & $T$ & $\ast$ & $\ast$ & $\ast$ \\
$F$     & $T$ & $T$   & $T$   & $T$   \\
\end{tabular}
\end{center}

Entries $\ast$ mark configurations where a $U_c$ operand renders the constraint uninterpretable at the value level; 
back-propagation abstains and logs no value-level event. 
Note $F$ dominates $U_c$ in $\land$, $T$ dominates $U_c$ in $\lor$, and $F$ dominates $U_c$ in $\to$ (vacuous implication). 
The strong-Kleene short-circuit conditions extended with $U_c$ as an abstaining fourth value. 
The $\to$ table is derivable from $\lor$ via $\lnot p \lor q$ under negation-as-identity on $U_c/U_k$. 
These tables suffice for $\mathrm{THEN}$-led constraints; $k$-ary generalizations compose associatively.

\section{Complexity Analysis}\label{app:complexity}

This appendix bounds the parsing, construction, verification, and path-analysis phases of the verifier. 
We prove general worst-case bounds (\thmref{total-worst}) and tighter bounds under the structural regularities of realistic CoT traces (the `mostly linear' regime;\thmref{total-linear}). 
It also introduces the pseudocode version of our verifier. 

\subsection{Preliminaries and Notation}\label{app:comp-prelim}

Following \secref{type6}, a CoT trace is a finite sequence $R = (s_1, \dots, s_n)$ with each
\[
s_i = \omega_{i,1}\,\omega_{i,2}\,\cdots\,\omega_{i,k_i}\;\pi_i \;\colon\; \varsigma_i,
\]
where $\omega_{i,j} \in \Omega = \Omega_u \cup \Omega_b \cup\{\mathrm{NOT}\}$, $\pi_i$ is an optional proposition in $\mathcal{P}$, and $\varsigma_i$ is an uninterpreted sentence. 
Operators partition into unary $\Omega_u = \{B, K, N, \mathord{?}, R\}$ and binary $\Omega_b =
\{\mathrm{THEN}, \mathrm{AND}, \mathrm{OR}, \mathrm{IF}\}$, with $\mathrm{NOT}$ a modifier. 

Our verifier takes in $n = |R|$ number of statements; $t$ maximum tokens in any single statement with $T = \sum_i (k_i + [\pi_i \neq \varnothing])$ total token count; and $p = |\mathcal{P}_R|$ distinct propositions.

Construction of the verifier yields a labelled multi-digraph $G = (V, E, \lambda)$ with $V = V_p \sqcup V_b$ (proposition nodes bijective with propositions, blank nodes separating consecutive operators). 
Every operator is an edge; parallel edges are preserved. 
Each $e$ carries a label $\lambda(e) \in \Omega$ and a kind $\kappa(e) \in \{\textsf{logical}, \textsf{gap}, \textsf{jump}, \textsf{loopback}, \textsf{meander}, \textsf{unknown}\}$.

This graph is parametrised by $|V|, |E|$ (node and edge counts). 
We denote $\Delta = \max_v \deg^+(v)$ the maximum out-degree and $\mu = \max_{u,v} |\{k : (u,v,k) \in E\}|$ its maximum edge multiplicity. 
We also assign $L$ loopback budget (\texttt{max\_loopbacks}, default 1), $J$ number of revisit edges (\textsf{jump}, \textsf{loopback}, and binary-connective edges combined). 
Crucially, we also cap the enumeration at $K$ (\texttt{max\_paths}), and the sample count at $S$; $\ell$: sampling walk-length cap. 

\textbf{Edge operations.} To build and traverse the graph, there are six edge insertion rules. 
Let $s_i$ be the current statement, and $u$ the node reached after $s_{i-1}$.
The operations are then: 

\begin{itemize}[leftmargin=*]
\item\textbf{Logical.} An explicit known operator produces a \textsf{logical}
edge (adjacent $\mathrm{NOT}$ folded in).
\item\textbf{Gap.} Two consecutive statements with no binary connector, whose target introduces \emph{new} material, are joined by a \textsf{gap}. 
Gaps are neutral.
\item\textbf{Jump.} Same condition as gap, but the target \emph{revisits} a seen proposition.
\item\textbf{Loopback.} Emitted by $R$; target is the proposition of the next statement carrying one. 
Chained $R$s compose serially.
\item\textbf{Meander.} Emitted by $N$; marks a pivot, enforcing a gap into $N$ but not out of it when the following statement carries operators. 
\item\textbf{Unknown.} Any unrecognized token becomes a labelled \textsf{unknown} edge, with a warning.
\end{itemize}

The \textsf{gap}/\textsf{jump} distinction is central to the mostly-linear analysis: both arise without a connective, but gap introduces (forward, fresh material) while jump returns (backward, prior material).

\subsection{General Worst-Case Bounds}\label{app:comp-worst}

In this section we prove the following: 

\begin{theorem}[Worst-case total]
\label{thm:total-worst}
The full pipeline runs in
\[
T_{\mathrm{total}} = O(nt) + T_{\mathrm{verify}} + T_{\mathrm{paths}} +
  O(N \cdot \bar p)
\]
steps, with $T_{\mathrm{verify}} = O(n^2 t)$, $T_{\mathrm{paths}}$ either
$T_{\mathrm{enum}}$ or $T_{\mathrm{sample}}$, $N$ the number of scored paths, and $\bar p$
the mean path length; space $S_{\mathrm{total}} = O(nt + N \cdot \bar p)$.
\end{theorem}

For this we need to prove parsing and construction costs (\appref{comp-build}), and the cost of path enumeration and sampling (\appref{comp-enum-sampling}); as well as back-propagation (\appref{comp-verifier}).

\subsubsection{Parsing and Construction}\label{app:comp-build}

We prove the following about graph construction (\algref{build}, $\textsc{BuildGraph}$) and its parsing subroutine (\algref{parse}; $\textsc{ParseLine}$):

\begin{lemma}[Graph construction]\label{lem:build}
$\textsc{BuildGraph}$ (\algref{build}) runs in $O(nt)$ time (excluding verifier work,
bounded in \appref{comp-verifier}) with $|V|, |E| \in O(nt)$. 

Moreover, when $t = O(1)$--the typical CoT regime--construction is $O(n)$ time
and space, with $|V|, |E| \in O(n)$.
\end{lemma}

\begin{proof}
For this, remark that $\textsc{ParseLine}$ runs in $O(t)$ and produces at most $t$ operators. 
Over $n$ statements, this provides the necessary bound. 

A minor nuance is that the adds (e.g., lines 6, 11, 14, etc.) are also constant time and add $\leq 1$ blank and $\leq 1$ edge. 
Given the operator list bound of $O(t)$, each statement contributes $O(t)$ nodes and edges and $O(t)$ work (with $O(1)$ bookkeeping for $\text{seen}$, $\text{meander}$, $\text{Rtail}$, source/sink). 
Summing gives $O(nt)$ and $|V|, |E| \in O(nt)$; post-loop source/sink resolution is $O(1)$.
\end{proof}

\begin{algorithm}[t]
\caption{$\textsc{ParseLine}(\ell)$}\label{alg:parse}
\begin{algorithmic}[1]
\State split $\ell$ on \texttt{`:'} into $(\text{left}, \varsigma)$
\State tokenize $\text{left}$ into $\tau_1, \dots, \tau_k$
\If{$\tau_k$ ends in \texttt{`?'} and $\tau_k \notin \Omega$}
  \State strip trailing \texttt{`?'}; append $\mathord{?}$ to the op list
\EndIf
\If{$\tau_k \in \Omega$}
  \State $\pi \gets \varnothing$; ops $\gets \tau_{1..k}$
\Else
  \State $\pi \gets \tau_k$; ops $\gets \tau_{1..k-1}$
\EndIf
\State fold each $\mathrm{NOT}$ into its preceding operator
\State \Return $(\text{ops}, \pi, \varsigma)$
\end{algorithmic}
\end{algorithm}

\begin{algorithm}[t]
\caption{$\textsc{BuildGraph}(R)$}\label{alg:build}
\begin{algorithmic}[1]
\State initialize empty $G$; $\text{prev} \gets \varnothing$;
$\text{seen} \gets \emptyset$
\State $\text{meander} \gets \bot$; $\text{Rtail} \gets \varnothing$;
initialize verifier state $\mathrm{vs}$
\For{$i \in 1..n$}
  \State $(\text{ops}, \pi, \varsigma) \gets \textsc{ParseLine}(s_i)$
  \If{$\pi \notin \text{seen}$ and $\pi \neq \varnothing$}
    \State add node $\pi$; $\text{seen} \gets \text{seen} \cup \{\pi\}$;
      $\text{revisit} \gets \bot$
  \ElsIf{$\pi \neq \varnothing$}
    \State $\text{revisit} \gets \top$
  \EndIf
  \If{$\text{Rtail} \neq \varnothing$ and $\pi \neq \varnothing$}
    \State add loopback $\text{Rtail} \to \pi$; $\text{Rtail} \gets
      \varnothing$; $\text{revisit} \gets \top$
  \EndIf
  \If{$s_i$ is a single-$R$ statement and $\text{Rtail} \neq \varnothing$}
    \State add blank $b'$; add loopback $\text{Rtail} \to b'$;
      $\text{Rtail} \gets b'$; \textbf{continue}
  \EndIf
  \If{$\text{ops} = \varnothing$} \Comment{bare proposition}
    \If{$\text{revisit}$} add jump $\text{prev} \to \pi$
    \Else\ add gap $\text{prev} \to \pi$ \EndIf
  \Else \Comment{statement carries operators}
    \State $\text{needsGap} \gets (\text{prev} \neq \varnothing) \land
      (\text{ops}_1 \notin \Omega_b) \land \neg\text{revisit} \land
      \neg\text{meander}$
    \State $c \gets \text{prev}$
    \If{$\text{needsGap}$} add blank $b$; add gap $\text{prev} \to b$;
      $c \gets b$
    \ElsIf{$\text{revisit}$} add blank $b$; add jump $\text{prev} \to b$;
      $c \gets b$ \EndIf
    \For{$j \in 1..|\text{ops}|$}
      \If{$j = |\text{ops}|$ and $\pi = \varnothing$ and $\text{ops}_j = R$}
        \State $\text{Rtail} \gets c$; \textbf{break} \EndIf
      \State $t \gets \pi$ if last op and $\pi \neq \varnothing$, else new blank
      \State add edge $c \xrightarrow{\text{ops}_j} t$ of appropriate kind;
        $c \gets t$
    \EndFor
  \EndIf
  \State \textsc{VerifierStep}($\mathrm{vs}$, ops, $\pi$, $\alpha_i$)
    \Comment{\appref{comp-verifier}}
  \State $\text{prev} \gets \pi$ if $\pi \neq \varnothing$ else $c$
  \State $\text{meander} \gets (\text{ops} \neq \varnothing \land
    \text{ops}_{\text{last}} = N)$
\EndFor
\State resolve source and sink; fail if unresolved
\State \Return $(G, \mathrm{vs})$
\end{algorithmic}
\end{algorithm}

\subsubsection{Path Enumeration and Sampling}\label{app:comp-enum-sampling}

In this section we prove the following:
\begin{lemma}[Enumeration cost]\label{lem:enum}
Let $N_{\mathrm{emit}} = \min(K, N_{\mathrm{total}})$. Then
$\textsc{EnumeratePaths}$ (\algref{enum}) runs in
\[
T_{\mathrm{enum}} = O\!\left(N_{\mathrm{emit}} \cdot (|V| + L + J) \cdot
  \Delta + W_{\mathrm{dead}}\right),
\]
with $W_{\mathrm{dead}}$ the frame work on non-yielding branches, in
$O(|V| + L + J)$ auxiliary space.
\end{lemma}

In Type-6's verifier, coherence is measured over \emph{edge-simple} source-to-sink paths (\algref{enum}). 
Parallel edges yield distinct paths, and three policies constrain enumeration. 
These are:
\begin{enumerate}[leftmargin=*]
    \item \textbf{Sink absorption}: once the sink is reached,
its out-edges are ignored
\item\textbf{Loopback budget}: at most $L$
loopback edges per path
\item\textbf{Restricted revisits}: a node may be re-entered only via a \textsf{jump}, \textsf{loopback}, or $\Omega_b$-labelled edge, never via a forward \textsf{logical} or \textsf{gap} edge.
\end{enumerate}

Policy (3) preserves legitimate revisits (jumps, loopbacks, connective returns) as distinct routes while pruning spurious forward-edge revisits, resolving both false-positive enumeration (double-counting sequential progress) and false-negative pruning.

Hence we need to account for these ahead of proving \lemref{enum}: 

\begin{lemma}[DFS frame cost and path length]
\label{lem:dfs}
A single $\textsc{DFS}$ frame at $v$ in \algref{enum}, excluding recursion, performs $O(\deg^+(v))$ work. 
Moreover, every emitted path has length at most $|V| - 1 + L + J$, so recursion
depth is $O(|V| + L + J)$, where $J$ counts \textsf{jump} and $\Omega_b$-labe;led edges.
\end{lemma}

\begin{proof}
A frame at $v$ iterates over $\deg^+(v)$ out-edges, each with $O(1)$ membership tests, kind/label lookups, a budget comparison, and set insert/removal. 
Moreover, $\textsc{DFS}$ refuses to re-enter a node in $V_v$ unless it is the sink or the edge is \textsf{jump}, \textsf{loopback}, or $\Omega_b$-labeled. 
Non-sink fresh-node entries are $\leq |V| - 1$; 
loopback traversals are capped at $L$, and jump/connective traversals are edge-simple, so $\leq J$. 
Hence $|\pi| \leq |V| - 1 + L + J$.
\end{proof}

\begin{lemma}[Simple-path count]\label{lem:paths}
The number of edge-simple source-to-sink paths in a multi-digraph on
$|V|$ nodes with maximum multiplicity $\mu$, respecting the three
policies, is
\[
N_{\mathrm{total}} \leq \mu^{|V| - 1 + L + J} \cdot
  \sum_{k=0}^{|V| - 2 + L + J} \frac{(|V| + L + J - 2)!}{k!}
= O\!\left(\mu^{|V| + L + J} \cdot (|V| + L + J)! \cdot e\right).
\]
\end{lemma}

\begin{proof}
By \lemref{dfs}, any path has length $\leq |V| - 1 + L + J$. 
In the underlying simple digraph (parallel edges collapsed), a simple path of edge-length $k+1$ is an ordered choice of $k$ distinct intermediates from $|V| + L + J - 2$ positions, $\leq (|V| + L + J - 2)!/k!$ choices; 
summing over $k$ gives the numerator, bounded by $(|V| + L + J - 2)! \cdot e$. Restoring multiplicities, each of $\leq |V| - 1 + L + J$ transitions selects one of $\mu$ parallel edges, giving $\mu^{|V| - 1 + L + J}$. 
Multiplying gives the claim.
\end{proof}

\textbf{Loopback budget}. Remark that $L$ does not reduce the leading factorial: an adversarial graph realizes the bound with paths containing no loopback edges, so the budget is never engaged. 
It does cap the length contributed by loopback revisits (\lemref{dfs}) and is essential to termination when loopback edges are present.

\begin{lemma}[Sampling cost]
\label{lem:sample}
$\textsc{SamplePaths}$ (\algref{sample}) runs in $O(S \cdot \ell_{\max}^2 \cdot \Delta)$ time and $O(\ell_{\max})$ auxiliary space per walk.
\end{lemma}

\begin{proof}
Each walk runs $\leq \ell_{\max}$ steps; each step iterates $\leq \Delta$ out-edges, testing membership against a partial path of length $\leq \ell_{\max}$ ($O(\ell_{\max})$), with $O(1)$ policy and budget checks and $O(1)$ selection. 
Per-step cost $O(\Delta \cdot \ell_{\max})$, per-walk $O(\ell_{\max}^2 \cdot \Delta)$, $S$ walks $O(S \cdot \ell_{\max}^2 \cdot \Delta)$; 
the $50S$ attempt cap is a constant factor. 
Sampled paths are drawn proportional to the product of uniform local choices. 
This is a well-defined (but in general biased) estimator relative to uniform over $N_{\mathrm{total}}$.
\end{proof}

\begin{algorithm}[t]
\caption{$\textsc{EnumeratePaths}(G, \mathrm{src}, \mathrm{snk}, K, L)$}
\label{alg:enum}
\begin{algorithmic}[1]
\State $Y \gets 0$
\Procedure{DFS}{$u, E_v, V_v, \ell$}
  \If{$Y \geq K$} \Return \EndIf
  \If{$u = \mathrm{snk}$ and $E_v \neq \emptyset$}
    \State \textbf{yield} $E_v$; $Y \gets Y + 1$; \Return \EndIf
  \For{$(u, v, k) \in \text{out-edges}(u)$}
    \If{$(u, v, k) \in E_v$} \textbf{continue} \EndIf
    \If{$v \in V_v$ and $v \neq \mathrm{snk}$
        and $\kappa((u,v,k)) \notin \{\textsf{jump}, \textsf{loopback}\}$
        and $\lambda((u,v,k)) \notin \Omega_b$}
      \State \textbf{continue} \EndIf
    \State $\ell' \gets \ell + [\kappa((u,v,k)) = \textsf{loopback}]$
    \If{$\ell' > L$} \textbf{continue} \EndIf
    \State \Call{DFS}{$v, E_v \cup \{(u,v,k)\}, V_v \cup \{v\}, \ell'$}
  \EndFor
\EndProcedure
\State \Call{DFS}{$\mathrm{src}, \emptyset, \{\mathrm{src}\}, 0$}
\end{algorithmic}
\end{algorithm}

\begin{algorithm}[t]
\caption{$\textsc{SamplePaths}(G, \mathrm{src}, \mathrm{snk}, S, L, \ell_{\max})$}
\label{alg:sample}
\begin{algorithmic}[1]
\For{$s \in 1..S$}
  \State $u \gets \mathrm{src}$; path $\gets [\,]$; $V_v \gets
    \{\mathrm{src}\}$; $\ell \gets 0$
  \For{$t \in 1..\ell_{\max}$}
    \If{$u = \mathrm{snk}$ and path $\neq [\,]$}
      \State \textbf{yield} path; \textbf{break} \EndIf
    \State $F \gets$ feasible out-edges of $u$ under the three policies
    \If{$F = \emptyset$} \textbf{break} \EndIf
    \State draw $(u, v, k) \sim \mathrm{Unif}(F)$; append to path; update
      $V_v, \ell, u$
  \EndFor
\EndFor
\end{algorithmic}
\end{algorithm}

We are now ready to prove \lemref{enum}.

\begin{lemma}[Enumeration cost]\label{lem:enum2}
Let $N_{\mathrm{emit}} = \min(K, N_{\mathrm{total}})$. Then
$\textsc{EnumeratePaths}$ (\algref{enum}) runs in
\[
T_{\mathrm{enum}} = O\!\left(N_{\mathrm{emit}} \cdot (|V| + L + J) \cdot
  \Delta + W_{\mathrm{dead}}\right),
\]
with $W_{\mathrm{dead}}$ the frame work on non-yielding branches, in
$O(|V| + L + J)$ auxiliary space. 

Moreover, ignoring $K$, $N_{\mathrm{total}}$ satisfies \lemref{paths} and
\[
T_{\mathrm{enum}} = O\!\left(\mu^{|V| + L + J} \cdot (|V| + L + J)!
  \cdot (|V| + L + J) \cdot \Delta\right).
\]

\end{lemma}
\begin{proof}
From \lemref{dfs} we know that each yielded path has length $\leq |V| - 1 + L + J$, produced by a branch of that depth with $O(\Delta)$ per frame. 
Charging frame work to terminal yields gives the first term, with non-yielding work $W_{\mathrm{dead}}$. 
Auxiliary state is the recursion stack plus $E_v, V_v$, all of which are $O(|V| + L + J)$.
The worst-case bound follows. 
\end{proof}

\subsubsection{Verifier Cost}\label{app:comp-verifier}

The verifier alternates within construction: after each statement,
\textsc{VerifierStep} updates $\mathrm{vs}$, extends any constraint
chain, and runs back-propagation. 

In this section we prove the following: 

\begin{lemma}[Amortised verifier cost]\label{lem:verifier-amortized}
Across the trace, $\textsc{VerifierStep}$ costs
\[
T_{\mathrm{verify}} = O\!\left(nt + n \cdot |\mathcal{K}| \cdot
  \kappa_{\max}\right) = O(n^2 t),
\]
where $|\mathcal{K}| = O(n)$, $\kappa_{\max} \leq t$; when $t = O(1)$,
$O(n^2)$.
\end{lemma}

For this we rely on the following propositions:

\begin{proposition}[Update costs]\label{prop:vs-update}
For a statement of $k \leq t$ operators, the sequential update
(K-/B-Commit, Bare-Assert, Doubt, Pivot, Loopback, or a NOT variant)
runs in $O(t)$; quality-signal checks are $O(1)$ using the tracked
history. 

Moreover, chain updates (Bare-Seed, And/Or-Extend, If-Open, Then-Close,
N-preservation, R-reset) run in $O(k) \leq O(t)$.
\end{proposition}
\begin{proof}
Each modal/bare update touches one entry of $\Sigma$ via $O(1)$ lookup and appends a constant-size event; 
the NOT fold is resolved at parse time (\lemref{build}). 
Doubt/Pivot/Loopback flip a flag or leave $\Sigma$ unchanged in $O(1)$. 
Iterating the $k$ operators is $O(k) \leq O(t)$. 
Quality checks compare against the prior commitment on the same proposition, $O(1)$ from $\mathrm{vs}$.

For the chain updates, note that seed and extensions append a constant-size triple; If-Open resets after
$O(1)$ $\mathrm{terminate}$; Then-Close builds a constraint via $\mathrm{close}$ in $O(|\chi|)$, with $|\chi| \leq k$. 
Iterating the $k$ operators is $O(k)$.
\end{proof}

\begin{proposition}[Back-propagation per statement]\label{prop:backprop}
With $|\mathcal{K}_i|$ the constraint-set size after $s_i$ and
$\kappa_{\max}$ the maximum chain length, back-propagation at step $i$
runs in $O(|\mathcal{K}_i| \cdot \kappa_{\max})$.
\end{proposition}
\begin{proof}
Each $c \in \mathcal{K}_i$ is evaluated by a scan of its operand chain
($\leq \kappa_{\max}$, each lookup $O(1)$); the three-way case split and
$O(1)$ event emission with $O(1)$ deduplication check follow. Summing
over $c$ gives the claim.
\end{proof}

We are now ready to prove \lemref{verifier-amortized}: 
\begin{lemma}[Amortised verifier cost]\label{lem:verifier-amortized2}
Across the trace, $\textsc{VerifierStep}$ costs
\[
T_{\mathrm{verify}} = O\!\left(nt + n \cdot |\mathcal{K}| \cdot
  \kappa_{\max}\right) = O(n^2 t),
\]
where $|\mathcal{K}| = O(n)$, $\kappa_{\max} \leq t$; when $t = O(1)$,
$O(n^2)$.
\end{lemma}

\begin{proof}
Immediate from \propref{vs-update} and \propref{backprop}. 
The former gives $O(t)$ per statement, $O(nt)$ total. 
The latter gives $O(|\mathcal{K}_i| \cdot \kappa_{\max})$ per step; 
summing with $|\mathcal{K}_i| \leq |\mathcal{K}| = O(n)$ gives $O(n^2 t)$.
\end{proof}

\subsection{Proof of \thmref{total-worst}}
The proof follows: 
\begin{theorem}[Worst-case total]\label{thm:total-worst2}
The full pipeline runs in
\[
T_{\mathrm{total}} = O(nt) + T_{\mathrm{verify}} + T_{\mathrm{paths}} +
  O(N \cdot \bar p)
\]
steps, with $T_{\mathrm{verify}} = O(n^2 t)$, $T_{\mathrm{paths}}$ either
$T_{\mathrm{enum}}$ or $T_{\mathrm{sample}}$, $N$ the number of scored paths, and $\bar p$
the mean path length; space $S_{\mathrm{total}} = O(nt + N \cdot \bar p)$.
\end{theorem}
\begin{proof}
Parse/build: $O(nt)$ (\lemref{build}). 
Verify: $O(n^2 t)$ time, $O(n)$ space (\lemref{verifier-amortized}). 
Paths: $T_{\mathrm{paths}}$ time with $O(nt + L + J)$ (enum) or $O(\ell_{\max})$ 
(sample) auxiliary space, absorbed into $O(nt)$. 
Scoring/aggregation: $O(N \cdot \bar p)$. 
Summing gives the claims.
\end{proof}

\subsection{The Mostly-Linear Regime}\label{app:mostly-linear}

The deduplication policy does not change the per-statement worst case (each constraint is still evaluated per step). 
It, however, bounds total emitted events at $O(|\mathcal{K}|) = O(n)$ rather than $O(n^2)$. 
In practice constraint sets rarely reach their worst-case size. 
Thus, in this section we show that Type-6 verification can be made \textit{mostly} linear. 
First, we define the following:

\begin{definition}[Mostly-linear graph]\label{def:mostly-linear}
We say a graph $G$ is \emph{mostly-linear} if $E = E_{\mathrm{spine}} \sqcup E_{\mathrm{par}} \sqcup E_{\mathrm{back}}$, 
where $E_{\mathrm{spine}}$ is a forward directed path (the \emph{spine}) visiting every non-blank node at most once and covering all propositions in first-appearance order; 
$E_{\mathrm{par}}$ are parallel-edge duplicates of spine/par edges; 
and $E_{\mathrm{back}}$ are \emph{backward} edges (target precedes source in spine order). 
We write $B = |E_{\mathrm{back}}|$ (so $B \geq J$, with equality when every backward edge is a revisit edge).
\end{definition}

Thus, in this section we prove the following:

\begin{theorem}[Mostly-linear total]
\label{thm:total-linear}
For a mostly-linear graph from $n$ statements with $t, \mu, \Delta, L,
\kappa_{\max} = O(1)$ and $B = O(\log n)$,
\[
T_{\mathrm{total}} =
\begin{cases}
O(n + \min(K, n^c) \cdot n) & \text{(enumeration),} \\
O(n + S \cdot n^2) & \text{(sampling),}
\end{cases}
\]
for a constant $c$ depending on backward-edge sparsity, in $O(n + N \cdot
\bar p)$ space.
\end{theorem}

Observe that the spine order in a mostly-linear graph is exactly the \textsf{jump} and \textsf{loopback} edges, plus any $\Omega_b$-labelled edge returning to a prior proposition. 
Hence, the decomposition from \defref{mostly-linear} reflects construction: 
each statement extends the spine forward by $O(t)$ edges;  
parallel edges arise when a statement lands on an existing proposition via a new operator  (e.g., $\mathrm{AND}\,a$ on a prior $a$); 
backward edges arise only from $R$-chain closures (\textsf{loopback}), revisit prefixes (\textsf{jump}), or connective edges to an earlier proposition. 
Gap edges are always forward.

We now introduce two technical lemmas for general mostly-linear multi-digraphs.
\begin{lemma}[Path count, mostly-linear]\label{lem:paths-linear}
The number of edge-simple source-to-sink paths in a mostly-linear multi-digraph, respecting the three policies, is $N_{\mathrm{total}} = O(\mu^{|V|} \cdot 2^{B})$.
\end{lemma}

\begin{proof}
Fix any edge-simple path $\pi$. 
Being node-simple away from the sink, and with the spine in fixed order, the spine nodes $\pi$ touches form a monotone subsequence determined by the backward edges $\pi$ traverses:
each backward edge selects a jump-back point, and between traversals the walk runs strictly forward along a contiguous segment. 
Thus $\pi$ is fixed by (i) a subset $S \subseteq E_{\mathrm{back}}$ (ordering forced up to constants) and (ii) a choice of $\leq \mu$ parallel edges per transition. 
Bounding (i) by $2^B$ and (ii) by $\mu^{|V|}$ (under $L + J =O(|V|)$) gives the claim.
\end{proof}

\begin{lemma}[Spine length]\label{lem:spine}
In a mostly-linear graph from $n$ statements with per-statement bound
$t$, $|E_{\mathrm{spine}}| = O(nt)$ and $|V| = O(nt)$.
\end{lemma}
\begin{proof}
Each statement contributes $O(t)$ nodes and edges (\lemref{build}); 
the spine is a subset of forward edges, one per added node, so $|E_{\mathrm{spine}}| \leq |V| - 1 = O(nt)$.
\end{proof}

With both we can now show that $\textsc{EnumeratePaths}$ and $\textsc{SamplePaths}$ in the mostly-linear regime runs in polynomial (mostly linear!) time: 

\begin{lemma}[Enumeration, mostly-linear]\label{lem:enum-linear}
For a mostly-linear graph, \textsc{EnumeratePaths} with cap $K$ runs in
\[
T_{\mathrm{enum}} = O\!\left(\min(K, \mu^{|V|} \cdot 2^{B}) \cdot (|V| +
  L + J) \cdot \Delta + W_{\mathrm{dead}}\right).
\]
When $B = O(\log |V|)$ and $\mu = O(1)$ (the typical case) the path count is polynomial in $|V|$ and enumeration runs in polynomial time.
\end{lemma}
\begin{proof}
Substitute \lemref{paths-linear} into \lemref{enum}. 
Under these assumptions $\mu^{|V|} \cdot 2^B = O(|V|^c)$, so the leading term is polynomial in $|V|$, hence in $n$ by \lemref{spine}.
\end{proof}

\begin{lemma}[Sampling, mostly-linear]\label{lem:sample-linear}
With $\ell_{\max} = O(|V| + L + J)$, $\textsc{SamplePaths}$ runs in $T_{\mathrm{sample}} = O(S \cdot (|V| + L + J)^2 \cdot \Delta)$. 
When $L = O(1)$, $J = O(\log n)$, $\Delta = O(1)$, this is $O(S \cdot n^2 t^2)$, or $O(S \cdot n^2)$ when $t = O(1)$.
\end{lemma}
\begin{proof}
Straightforward, from substituting $\ell_{\max} = O(|V| + L + J)$ into \lemref{sample} and using \lemref{spine} for simplification. 
\end{proof}

\begin{lemma}[Dead-end work, mostly-linear]\label{lem:dead-linear}
For a mostly-linear graph, $W_{\mathrm{dead}} = O((|V| + L + J) \cdot
\Delta \cdot 2^{B})$.
\end{lemma}
\begin{proof}
A branch fails to yield only by exhausting feasible out-edges before the
sink. Forward spine edges either advance or terminate at a blank
($O(\Delta)$ work, then return). Dead-end branches thus correspond to
backward-edge choice sequences failing to re-enter the spine, $\leq 2^B$
of them; each has depth $O(|V| + L + J)$ (Lemma~\ref{lem:dfs}(ii)) with
$O(\Delta)$ per frame.
\end{proof}

\begin{lemma}[Verifier cost, mostly-linear]\label{lem:verify-linear}
On a mostly-linear graph with $B = O(\log n)$ and $\kappa_{\max} = O(1)$, $T_{\mathrm{verify}} = O(n)$.
\end{lemma}

\begin{proof}
Per-statement update and chain work are $O(t) = O(1)$ (\propref{vs-update}), totalling $O(n)$. 
Back-propagation costs $O(|\mathcal{K}_i| \cdot \kappa_{\max}) = O(|\mathcal{K}_i|)$ (\propref{backprop}); 
when prior constraints are largely inert. 
Also recall from earlier that the deduplication policy bounds total re-emitted events at $O(n)$.
\end{proof}

We are now ready to prove \thmref{total-linear}:

\begin{theorem}[Mostly-linear total]\label{thm:total-linear2}
For a mostly-linear graph from $n$ statements with $t, \mu, \Delta, L,
\kappa_{\max} = O(1)$ and $B = O(\log n)$,
\[
T_{\mathrm{total}} =
\begin{cases}
O(n + \min(K, n^c) \cdot n) & \text{(enumeration),} \\
O(n + S \cdot n^2) & \text{(sampling),}
\end{cases}
\]
for a constant $c$ depending on backward-edge sparsity, in $O(n + N \cdot
\bar p)$ space.
\end{theorem}

\begin{proof}
Parse/build $O(n)$ (\lemref{build}); 
verify $O(n)$ (\lemref{verify-linear}); 
paths $O(\min(K, n^c) \cdot n)$ (\lemref{enum-linear}, with $W_{\mathrm{dead}}$ absorbed by \lemref{dead-linear}) or $O(S \cdot n^2)$ (\lemref{sample-linear}); 
scoring $O(N \cdot \bar p)$ with $N = O(\min(K, n^c))$ or $S$ and $\bar p = O(n)$, absorbed into the leading term. 
Summing gives the claim.
\end{proof}

CoT traces of length $n = 20$--$100$ produce graphs with $|V| \approx n$, $\Delta \leq 4$, $\mu \leq 2$, and a handful of backward edges. 
Enumeration stays within the default cap $K = 10^4$, the coherence report is produced in time polynomial in trace length, and the verifier runs linearly.
Hence, in the practical regime the pipeline is polynomial in trace length and tractable for realistic CoT sizes.

\section{Extended Results}\label{app:extended-results}

\subsection{Hard-fail and trace-level scores per model}

Hard-fail rates broken down by model and dataset split are in \tabref{allmodel-split-hard-fail}. 
LINC's near-zero hard-fail rates (0-5\% across all generators and splits) are a definitional consequence of its formal-verification mechanism: 
LINC hard-fails only when the LLM-generated FOL translation is syntactically ill-formed, not when the reasoning is substantively wrong. 
This is a \textbf{very different notion of `failure'} from the LLM-judge or PRM methods, which flag reasoning-level errors, and motivates our reporting of both hard-fail agreement and scalar correlations rather than either alone. 
LLM-judge hard-fail rates cluster strongly by judge capability: 
capable judges (Claude Opus, GPT-5.6, Gemma-4B) hard-fail 50-70\% of traces, while smaller or less-aligned judges (Qwen-3.5-9B, Llama-3.1-8B, GLM-4.7) hard-fail 97-100\%. 
This spread suggests hard-fail as measured is at least as much a property of the LLM as of the traces themselves. 
Finally, the PRM methods we tested exhibit near-total hard-fail rates (90-100\%) on natural-language reasoning traces despite being trained for step-wise correctness signals in mathematics; 
this is a \textbf{cross-domain transfer failure} worth further study. 

In \tabref{allmodel-split-trace-score} we report mean trace-level scores over the same breakdown. 
As before, LLM-judges span a wide range of mean scores (2.0-4.2), and their rankings across generation models are not consistent. 
PRM methods produce mean scores from 0.22 to 4.32 depending on model, with the near-zero mean of Qwen-2.5-PRM reflecting its consistent near-boundary decisions. 
Type-6 sits between the LLM-judge and PRM regimes, at a mean of 1.34 overall with a notable spike on the MMLU split (2.06). 
The trivial baseline sits at 2.60 across all splits by construction, providing a fixed reference line: 
any method whose mean score is close to this line is not distinguishing high- from low-quality traces on average, though it may still be doing so per-trace.

\begin{table}[]
    \centering
    \small
    \begin{tabular}{lrrrrr}
Hard-fail rate & CoT-Logic & DebateLab & MMLU & Hand-crafted & All \\
\midrule
\textit{LINC} & & &  & & \\
gemma4\_e4b & 0.0\%\,(183) & 0.0\%\,(228) & 0.0\%\,(348) & 0.0\%\,(28) & 0.0\%\,(787) \\
gpt\_56 & 1.1\%\,(183) & 7.0\%\,(273) & 3.8\%\,(346) & 17.6\%\,(34) & 4.8\%\,(836) \\
qwen35\_9b & 0.0\%\,(225) & 0.0\%\,(257) & 0.0\%\,(336) & 0.0\%\,(35) & 0.0\%\,(853) \\
\midrule
\textit{LLM-as-a-judge} & & &  & & \\
opus48 & 69.6\%\,(276) & 84.2\%\,(311) & 51.4\%\,(370) & 84.6\%\,(39) & 68.0\%\,(996) \\
gemma4 & 62.1\%\,(256) & 54.2\%\,(299) & 66.3\%\,(362) & 64.9\%\,(37) & 61.3\%\,(954) \\
glm47 & 100.0\%\,(148) & 100.0\%\,(186) & 99.3\%\,(150) & 100.0\%\,(16) & 99.8\%\,(500) \\
gpt56 & 63.1\%\,(274) & 79.4\%\,(311) & 50.5\%\,(370) & 76.9\%\,(39) & 64.1\%\,(994) \\
llama31 & 98.0\%\,(100) & 99.0\%\,(103) & 97.5\%\,(161) & 100.0\%\,(15) & 98.2\%\,(379) \\
qwen35 & 99.6\%\,(227) & 100.0\%\,(256) & 95.4\%\,(304) & 94.6\%\,(37) & 97.9\%\,(824) \\
\midrule
\textit{PRM} & & &  & & \\
llama31-prm & 64.9\%\,(276) & 73.0\%\,(311) & 64.6\%\,(370) & 92.3\%\,(39) & 68.4\%\,(996) \\
math-shepherd & 98.6\%\,(276) & 100.0\%\,(311) & 90.8\%\,(370) & 100.0\%\,(39) & 96.2\%\,(996) \\
qwen25-prm & 100.0\%\,(276) & 100.0\%\,(311) & 100.0\%\,(370) & 100.0\%\,(39) & 100.0\%\,(996) \\
\midrule
\textit{ROSCOE} & & &  & & \\
deberta-large & 100.0\%\,(276) & 100.0\%\,(311) & 100.0\%\,(370) & 100.0\%\,(39) & 100.0\%\,(996) \\
\midrule
\textit{Trivial} & & &  & & \\
trivial & 0.0\%\,(276) & 0.6\%\,(311) & 0.0\%\,(370) & 0.0\%\,(39) & 0.2\%\,(996) \\
\midrule
\textit{Type-6} & & &  & & \\
type6 & 87.9\%\,(273) & 89.0\%\,(309) & 79.2\%\,(370) & 94.7\%\,(38) & 85.3\%\,(990) \\
    \end{tabular}
\caption{Hard-fail rates by generation model and evaluation split, per method. 
Each cell reports the fraction of traces in the split that the method flagged as failing, with the number of traces in parentheses. 
Higher rates indicate the method rejects more traces; 
they do not directly indicate accuracy, as ground truth in most splits is unknown. 
LINC (three models, top block) produces a formal verdict on translated first-order-logic hypotheses and rejects a trace only when the formal check itself fails, giving very low hard-fail rates by construction. 
Trivial (bottom block) is a uniform-random baseline calibrated to reject $\sim$0\% of traces. 
Between these extremes, LLM-judge, PRM, and Type-6 methods vary substantially across generation model and split, reflecting both genuine reasoning-quality variation across generators and each method's internal sensitivity to trace surface features.}
\label{tab:allmodel-split-hard-fail}
\end{table}

\begin{table}[]
    \centering
    \small
    \begin{tabular}{lrrrrr}
Mean trace score & CoT-Logic & DebateLab & MMLU & Hand-crafted & All \\
\midrule
linc:gemma4\_e4b & --- & --- & --- & --- & --- \\
linc:gpt\_56 & --- & --- & --- & --- & --- \\
linc:qwen35\_9b & --- & --- & --- & --- & --- \\
\midrule
llm\_judge:claude\_opus\_48 & $3.446$ & $3.068$ & $3.741$ & $3.026$ & $3.421$ \\
llm\_judge:gemma4\_e4b & $4.277$ & $4.174$ & $4.022$ & $4.108$ & $4.142$ \\
llm\_judge:glm\_47 & $3.209$ & $3.081$ & $2.713$ & $2.625$ & $2.994$ \\
llm\_judge:gpt\_56 & $3.507$ & $3.048$ & $3.711$ & $3.026$ & $3.421$ \\
llm\_judge:llama\_31\_8b & $3.640$ & $3.379$ & $2.957$ & $3.400$ & $3.269$ \\
llm\_judge:qwen35\_9b & $2.004$ & $1.941$ & $2.132$ & $2.000$ & $2.032$ \\
\midrule
prm:llama31-8b-prm & $4.374$ & $4.249$ & $4.367$ & $4.183$ & $4.325$ \\
prm:math-shepherd & $1.724$ & $1.470$ & $2.145$ & $1.602$ & $1.796$ \\
prm:qwen25-prm & $0.208$ & $0.207$ & $0.226$ & $0.224$ & $0.215$ \\
\midrule
roscoe:deberta-large & $1.414$ & $1.423$ & $1.199$ & $1.381$ & $1.336$ \\
\midrule
trivial:trivial & $2.593$ & $2.648$ & $2.559$ & $2.563$ & $2.596$ \\
\midrule
type6 & $0.850$ & $0.976$ & $2.062$ & $0.855$ & $1.342$ \\
\end{tabular}
\caption{Mean trace-level coherence score $\in [0, 5]$ by
generation model and split, per method. Each cell reports the
method-emitted or method-derived scalar averaged across traces in
the split. LINC rows are marked with `---' because LINC produces
a binary verdict rather than a scalar coherence score, and no
principled projection to $[0, 5]$ exists; agreement between LINC
and the scalar methods is reported through hard-fail agreement
(\tabref{hard-fail-agreement-correlation}) instead. 
}\label{tab:allmodel-split-trace-score}
 \end{table}

\subsection{Per-model analysis}\label{app:per-method-tables}

We discuss further details from the ablation study from \secref{per-baseline}. 

\paragraph{Trivial baseline.}
Trivial's four signals span the corpus with median statement count
$55.5$ (IQR $[41, 75]$) and typ--token ratio median $0.376$. The
non-trivial correlation with LLM-judges does not necessarily
indicate false positives on the judges' part: shorter or more
repetitive traces may be genuinely less coherent, and Trivial may
be picking up a real signal. 
\tabref{trivial-summary} reports the distribution of the four
Trivial signals. Ranges are broad on all four,
providing informative surface features despite the baseline's
by-construction low hard-fail rate.

\begin{table}[t]
\centering
\small
\begin{tabular}{lrrrrr}
\toprule
Signal & Min & Q1 & Median & Q3 & Max \\
\midrule
Statement count & 4 & 41.0 & 55.5 & 75.0 & 326 \\
Connective density & 0.000 & 0.183 & 0.247 & 0.316 & 0.772 \\
Repetition rate & 0.000 & 0.022 & 0.046 & 0.089 & 0.386 \\
Type-token ratio & 0.080 & 0.301 & 0.376 & 0.437 & 0.740 \\
\bottomrule
\end{tabular}
\caption{Distributional summary of the four Trivial signals.}
\label{tab:trivial-summary}
\end{table}
\subsubsection{ROSCOE}
\label{sec:roscoe-diagnostics}

\paragraph{ROSCOE.}
Individual ROSCOE metrics carry informative distributions: chain
self-consistency median $0.056$ (IQR $[0.042, 0.070]$); CSE-step
(contradiction) median $0.033$ (IQR $[0.018, 0.063]$). The
degenerate hard-fail flag is therefore an artefact of ROSCOE's
default threshold rather than of its underlying signals; the
threshold was calibrated on datasets other than natural-language
reasoning and does not transfer here. 
\tabref{roscoe-summary} reports the distribution of the six ROSCOE
metrics used to compute the ROSCOE trace-level scalar and hard-fail
flag. %

\begin{table}[t]
\centering
\small
\begin{tabular}{lrrrrr}
\toprule
Metric & Min & Q1 & Median & Q3 & Max \\
\midrule
Faithfulness-step & 0.037 & 0.147 & 0.264 & 0.360 & 0.538 \\
Informativeness-step & 0.163 & 0.251 & 0.284 & 0.322 & 0.508 \\
Repetition-step & 0.281 & 0.559 & 0.600 & 0.651 & 0.863 \\
Reasoning alignment & 0.040 & 0.299 & 0.397 & 0.510 & 0.837 \\
Chain self-consistency & 0.001 & 0.042 & 0.056 & 0.070 & 0.313 \\
CSE-step (contradiction) & 0.002 & 0.018 & 0.033 & 0.063 & 0.432 \\
\bottomrule
\end{tabular}
\caption{Distributional summary of ROSCOE metrics.}
\label{tab:roscoe-summary}
\end{table}

\paragraph{LINC.}
\tabref{linc-summary} reports corpus-wide LINC translation-fidelity
distributions and verdict counts. 
Predicate grounding is at $1.000$ across all three LINC models,
suggesting that ontology extraction is a stable operation even
when the source-to-logic translation is not. 
\tabstworef{linc-verdicts-per-model}{linc-fidelity-per-model}
disaggregate these by underlying translation model. As discussed in
the main text, only GPT-5.6 produces informative
entailment/non-entailment verdicts on this corpus; Qwen-3.5-9B and
GLM-4.7 emit unknown or error verdicts on nearly all traces despite
Qwen achieving the highest coverage ratio of the three. 
All models showed zero within-method agreement (omitted).
The observation that GPT-5.6 achieves the lowest coverage ratio
but the only informative verdict distribution suggests that when
LINC's translation is confident enough to complete, it is also
confident enough to commit to a verdict; when translation is
partial, LINC reserves judgement and emits `unknown'.

\begin{table}[t]
\centering
\small
\begin{tabular}{lrrrrr}
\toprule
Fidelity metric & Min & Q1 & Median & Q3 & Max \\
\midrule
Coverage ratio & 0.001 & 0.221 & 0.329 & 0.431 & 1.000 \\
Predicate grounding & 0.000 & 1.000 & 1.000 & 1.000 & 1.000 \\
Parse errors per trace & 0 & 0.0 & 1.0 & 4.0 & 61 \\
Undeclared predicates & 0 & 0.0 & 0.0 & 0.0 & 60 \\
\midrule
Verdict & Count & \% of corpus &&& \\
\midrule
Entailment & 38 & 3.8\% & & & \\
Non-entailment & 189 & 19.0\%  & & & \\
Contradiction (hard-fail) & 14 & 1.4\% & && \\
Unknown & 746 & 74.9\% & && \\
Error (parse, timeout, or translation failure) & 9 & 0.9\% & && \\
\bottomrule
\end{tabular}
\caption{LINC summaries of translation fidelity (top) and verdict distributions (bottom). 
Coverage ratio is the fraction of source premises translated; predicate grounding the
fraction of predicates in the trace-extracted ontology. 
For verdicts, the hard-fail flag fires on contradiction.
}
\label{tab:linc-summary}
\end{table}

\begin{table}[h]
\centering
\small
\begin{tabular}{lrrrrr}
\toprule
Model & Entailment & Non-entail. & Contradiction & Unknown & Error \\
\midrule
GPT-5.6      & 14.7\% & 65.3\% & 4.0\% & 0.0\%  & 16.1\% \\
Gemma-4-E4B  & 0.0\% & 0.0\%  & 0.0\% & 79.0\% & 21.0\% \\
Qwen-3.5-9B  & 0.0\%  & 0.0\%  & 0.0\% & 85.6\% & 14.4\% \\
GLM-4.7      & 0.0\%  & 0.0\%  & 0.0\% & 72.8\% & 27.2\% \\
\bottomrule
\end{tabular}
\caption{LINC verdicts by underlying translation model. 
Only GPT-5.6 produces informative entailment/non-entailment verdicts;
smaller models emit predominantly unknown or error.}
\label{tab:linc-verdicts-per-model}
\end{table}

\begin{table}[h]
\centering
\small
\begin{tabular}{lrrrrr}
\toprule
Model & Min & Q1 & Median & Q3 & Max \\
\midrule
\multicolumn{6}{l}{\emph{Coverage ratio}} \\
GPT-5.6    & 0.000 & 0.004 & 0.006 & 0.008 & 0.018 \\
Qwen-3.5-9B& 0.033 & 0.328 & 0.515 & 0.714 & 1.024 \\
GLM-4.7    & 0.001 & 0.006 & 0.010 & 0.014 & 0.049 \\
Gemma-4-E4B  & 0.034 & 0.311 & 0.469 & 0.614 & 1.000 \\
\midrule
\multicolumn{6}{l}{\emph{Parse errors per trace}} \\
GPT-5.6    & 0 & 1.0 & 2.0 & 4.0 & 38 \\
Qwen-3.5-9B& 0 & 0.0 & 0.0 & 1.0 & 61 \\
GLM-4.7    & 0 & 0.0 & 0.0 & 0.0 & 161 \\
Gemma-4-E4B  & 0 & 0.0 & 0.0 & 3.0 & 36 \\
\bottomrule
\end{tabular}
\caption{LINC translation fidelity by underlying model. Coverage
ratio counts fraction of source premises formalised; higher is
better. GPT-5.6 emits the most parse errors but achieves the
lowest coverage ratio, indicating it fails on hard-to-translate
premises rather than surface parsing.}
\label{tab:linc-fidelity-per-model}
\end{table}

\paragraph{PRMs.}
The progressive collapse of `pass' emissions from Llama
($68\%$ hard-fail) through Math-Shepherd ($96\%$) to Qwen
($100\%$) suggests that transfer to natural-language reasoning
degrades sharply with the specificity of the PRM's mathematical
training distribution. Math-Shepherd's non-zero low-confidence
rate ($5.3\%$) provides a modest reference point: when the PRM
has any signal of uncertainty at all, it produces it; when it
does not, as with Llama and Qwen at $0\%$ low-confidence,
mislabelling is confident.
\tabref{prm-summary} reports the corpus-wide distribution of
per-trace PRM aggregations, and \tabref{prm-operational-per-model}
disaggregates the operational signals by PRM model. 

\begin{table}[t]
\centering
\small
\begin{tabular}{lrrrrr}
\toprule
Signal & Min & Q1 & Median & Q3 & Max \\
\midrule
Trace mean score & 0.267 & 0.392 & 0.426 & 0.453 & 0.564 \\
Trace minimum score & 0.005 & 0.037 & 0.067 & 0.197 & 0.454 \\
Trace variance & 0.0009 & 0.0087 & 0.0269 & 0.0404 & 0.0754 \\ 
\midrule
Signal & Traces & \% of corpus \\
\midrule
Low-confidence (var $<$ 0.01, mean $\in [0.4, 0.6]$) & 0 & 0.0\% & \\
Truncated (steps did not fit context) & 39 & 3.9\% &\\
Hard-fail (mean $<$ 0.5 or min $<$ 0.2) & 967 & 97.1\% &\\
\bottomrule
\end{tabular}
\caption{Distributional summary of PRM per-trace aggregations (top) and operational signals (bottom). 
}
\label{tab:prm-summary}
\end{table}

\begin{table}[h]
\centering
\small
\begin{tabular}{lrrr}
\toprule
PRM & Low-confidence & Truncated & Hard-fail \\
\midrule
Llama-3.1-8B-PRM   & 0.0\%  & 2.7\% & 68.4\% \\
Math-Shepherd      & 5.3\%  & 3.6\% & 96.2\% \\
Qwen-2.5-PRM       & 0.0\%  & 4.0\% & 100.0\% \\
\bottomrule
\end{tabular}
\caption{PRM operational signals by model. Low-confidence signals
low-variance predictions near the boundary (var $<$ 0.01, mean $\in [0.4, 0.6]$); 
hard-fail uses mean $< 0.5$ or minimum $< 0.2$.}
\label{tab:prm-operational-per-model}
\end{table}

\paragraph{LLMs-as-judges.}
Per-criterion flag rates reveal the underlying instability of
LLM-judge measurements. `Unresolved doubt' is flagged in
$52$-$87\%$ of traces depending on judge, and `modal mismatch'
in $29$-$95\%$; no single criterion produces stable rates
across judges. This confirms that the judge-model effect is not
narrowly attributable to overall scoring calibration but persists
criteria by criteria in the rubric. 
\tabref{llmjudge-flag-rates} reports the corpus-wide flag rate for
each of the four rubric axes (contradiction, unsupported conclusion,
modal mismatch, unresolved doubt).
\tabref{llmjudge-per-model} disaggregates these by judge model.
Even on this per-criterion rubric, judges vary by two to three times in
how often they flag a given criterion on the same corpus.
\tabref{within-llmjudge-kappa} shows the resulting within-method
pairwise agreement, with the maximum below $\kappa = 0.4$.

\begin{table}[t]
\centering
\small
\begin{tabular}{lrr}
\toprule
Axis & Flag rate & $\kappa$ with Type-6 analog \\
\midrule
Contradiction & 33.7\% & --- \\
Unsupported conclusion & 62.0\% & --- \\
Modal mismatch & 60.7\% & --- \\
Unresolved doubt & 80.3\% & --- \\
\bottomrule
\end{tabular}
\caption{LLM-judge per-axis flag rates and agreement with the
analogous Type-6 category: contradiction $\leftrightarrow$ $K$-$K$ plus
derived; unsupported conclusion $\leftrightarrow$ unverifiable
derivation plus reasoning-avoidance; modal mismatch $\leftrightarrow$
modal mismatch $(U_k/K)$; unresolved doubt $\leftrightarrow$ residual
unresolved-doubt at sink.}
\label{tab:llmjudge-flag-rates}
\end{table}

\begin{table}[h]
\centering
\small
\begin{tabular}{lrrrr}
\toprule
Judge & Contradict. & Unsup. concl. & Modal mismatch & Unresolved doubt \\
\midrule
Claude Opus 4.8   & 23.7\% & 54.7\% & 29.2\% & 66.3\% \\
Gemma-4-E4B       & 2.6\%  & 28.8\% & 49.8\% & 59.3\% \\
GLM-4.7           & 45.2\% & 99.4\% & 38.8\% & 52.6\% \\
GPT-5.6           & 34.3\% & 34.2\% & 43.3\% & 87.1\% \\
Llama-3.1-8B      & 33.5\% & 38.0\% & 95.0\% & 72.8\% \\
Qwen-3.5-9B       & 73.9\% & 95.9\% & 89.4\% & 85.2\% \\
\bottomrule
\end{tabular}
\caption{LLM-judge per-axis flag rates by judge model. Rates
vary by two to three times across judges on the same rubric axis
and the same corpus.}
\label{tab:llmjudge-per-model}
\end{table}

\paragraph{Type-6.}
Type-6's event incidence table shows that soft contradictions
(B-B conflicts, unjustified downgrades) and residuals at sink are
common, but only elevate to hard failures when they touch the
direct implication chain. This design decision is what allows
Type-6 to report a low false-positive rate: most of the corpus is
context or exposition, and a naive contradiction-counter would
hard-fail nearly everything. 
\tabref{type6-event-incidence} reports the incidence of each
Type-6 event category across the corpus, with the number of traces
triggered and average events per triggered trace.
\tabref{type6-chain-stats} reports the direct-implication-chain
statistics used to elevate residual signals to hard-fails.
\tabref{type6-reductions} reports the Spearman correlations
between three alternative graded-reduction strategies for the
Type-6 scalar, all of which are strongly rank-correlated.

\begin{table}[h]
\centering
\small
\begin{tabular}{lrrrrrr}
\toprule
 & Claude Opus & Gemma-4 & GLM-4.7 & GPT-5.6 & Llama-3.1 & Qwen-3.5 \\
\midrule
Claude Opus & ---     &         &         &         &         &     \\
Gemma-4     & $0.189$ & ---     &         &         &         &     \\
GLM-4.7     & $0.008$ & $0.006$ & ---     &         &         &     \\
GPT-5.6     & $0.382$ & $0.280$ & $0.007$ & ---     &         &     \\
Llama-3.1   & $-0.010$& $0.011$ & $0.000$ & $0.011$ & ---     &     \\
Qwen-3.5    & $0.084$ & $0.071$ & $0.196$ & $0.074$ & $0.090$ & --- \\
\bottomrule
\end{tabular}
\caption{Within-LLM-judge pairwise Cohen's $\kappa$ on hard-fail.
Even the highest agreement ($\kappa = 0.382$ between Claude Opus
and GPT-5.6) is below moderate.}
\label{tab:within-llmjudge-kappa}
\end{table}

\begin{table}[t]
\centering
\small
\begin{tabular}{lrrr}
\toprule
Category & Traces triggered & Avg.\ events/triggered & Blocks path? \\
\midrule
\multicolumn{4}{l}{\emph{Hard contradictions}} \\
$K$-$K$ contradiction & 62 & 2.258 & \checkmark \\
Derived contradiction & 409 & 5.276 & \checkmark \\
Modal mismatch $(U_c)$ & 36 & 1.417 & \checkmark \\
Modal mismatch $(U_k / K)$ & 182 & 2.308 & \checkmark \\
Reasoning-avoidance & 548 & --- & \checkmark (trace) \\
\midrule
\multicolumn{4}{l}{\emph{Soft contradictions}} \\
$B$-$B$ conflict & 87 & 2.805 &  \\
Bare-reassertion conflict & 68 & 2.206 &  \\
Unjustified downgrade & 200 & 2.225 &  \\
Unjustified modal shift & 75 & 1.680 &  \\
\midrule
\multicolumn{4}{l}{\emph{Residual (at sink)}} \\
Unresolved doubt & 495 & --- & \\
Unresolved unknowability & 878 & --- & \\
\midrule
\multicolumn{4}{l}{\emph{Quality signals}} \\
Redundant re-assertion & 171 & 2.000 &  \\
Self-questioned $K$ & 0 & --- &  \\
Self-questioned $B$ & 14 & 2.000 &  \\
Unverifiable derivation & 250 & 2.716 &  \\
Ambiguous negated connective & 72 & 2.792 &  \\
\bottomrule
\end{tabular}
\caption{Incidence of each Type-6 event category. `Traces triggered'
counts traces with at least one event of that category; the average is
the mean event count within those traces. Residual signals at sink are
elevated to hard when on the direct implication chain
(\secref{type6}).}
\label{tab:type6-event-incidence}
\end{table}

\begin{table}[t]
\centering
\small
\begin{tabular}{lrrrrr}
\toprule
Statistic & Min & Q1 & Median & Q3 & Max \\
\midrule
Chain length (props on direct implication chain to sink)
                   & 1 & 1.0 & 1.0 & 3.0 & 32 \\
On-chain fraction (chain length $/$ trace length)
                   & 0.003 & 0.016 & 0.025 & 0.074 & 0.733 \\
Elevated residuals per trace
                   & 0 & 0.0 & 0.0 & 0.0 & 16 \\
\bottomrule
\end{tabular}
\caption{Direct-implication-chain statistics. Higher on-chain fractions
indicate tighter reasoning; elevated residuals count the
unresolved-doubt or unresolved-$U_k$ signals raised to hard by touching
the chain.}
\label{tab:type6-chain-stats}
\end{table}

\begin{table}[h]
\centering
\small
\begin{tabular}{lrrr}
\toprule
 & strict & graded & proportional \\
\midrule
strict       & ---     &         &     \\
graded       & $0.703$ & ---     &     \\
proportional & $0.697$ & $0.976$ & --- \\
\bottomrule
\end{tabular}
\caption{Spearman $\rho$ between Type-6 alternative reduction
strategies on the trace-level coherence scalar. All three are
strongly rank-correlated, indicating the choice of reduction is
not a large source of variance in Type-6's output.}
\label{tab:type6-reductions}
\end{table}

\end{document}